\documentclass[journal]{IEEEtran} 
\usepackage[utf8]{inputenc}
\usepackage{fullpage} 
\usepackage{url}
\usepackage{subcaption}
\usepackage{graphicx}
\usepackage{xcolor}
\usepackage{soul}  
\usepackage{amsmath}
\usepackage{amssymb}
\usepackage{amsfonts}
\usepackage{amsthm}
\usepackage{algorithm}
\usepackage{algorithmic}

\usepackage{latexsym,mathrsfs}
\usepackage{enumerate,verbatim}

\usepackage{mathtools}

\def\modif{\color{black}}
\def\endmodif{\color{black}}

\newtheorem{lemma}{Lemma}

\newtheorem{theorem}{Theorem}
\newtheorem{corollary}{Corollary}
\newtheorem{definition}{Definition}
\newtheorem*{definition*}{Definition}

\newtheorem{remark}{Remark}

\DeclareMathOperator{\conv}{conv}

\DeclareMathOperator{\cone}{cone} 
\DeclareMathOperator{\rank}{rank}

\DeclareMathOperator{\cond}{cond} 
\DeclareMathOperator{\logdet}{logdet} 
\DeclareMathOperator{\tr}{trace} 
\DeclareMathOperator{\diag}{Diag} 

\makeatletter
\renewcommand{\maketag@@@}[1]{\hbox{\m@th\normalsize\normalfont#1}}%
\makeatother

\begin{document}

\title{Blind Audio Source Separation with \\ 
Minimum-Volume Beta-Divergence NMF}

\author{~Valentin Leplat, Nicolas Gillis, Andersen M.S. Ang*\thanks{* Department of Mathematics and Operational Research,
Facult\'e Polytechnique, Universit\'e de Mons,
Rue de Houdain 9, 7000 Mons, Belgium. Authors acknowledge the support by the Fonds de la Recherche Scientifique - FNRS and the Fonds Wetenschappelijk Onderzoek - Vlanderen (FWO) under EOS Project no O005318F-RG47, and by the European Research Council (ERC starting grant no 679515). 
E-mails: \{valentin.leplat, nicolas.gillis, manshun.ang\}@umons.ac.be. 

Manuscript received in July 2019. Accepted April 2020.} 
}
\markboth{IEEE Transactions on Signal Processing ,~Issue XX, Month~2020}%
{Shell \MakeLowercase{LEPLAT, GILLIS, ANG: Blind Audio Source Separation with Min-Vol NMF}}

\maketitle
\begin{abstract} 
Considering a mixed signal composed of various audio sources and recorded with a single microphone,  we consider on this paper the blind audio source separation problem which consists in isolating and extracting each of the sources. To perform this task, nonnegative matrix factorization (NMF) based on the Kullback-Leibler and Itakura-Saito $\beta$-divergences is a standard and state-of-the-art technique that uses the time-frequency representation of the signal. We present a new NMF model better suited for this task. It is based on the minimization of $\beta$-divergences along with a penalty term that promotes the columns of the dictionary matrix to have a small volume. Under some mild assumptions and in noiseless conditions, we prove that this model is provably able to identify the sources. In order to solve this problem, we propose multiplicative updates whose derivations are based on the standard majorization-minimization framework. We show on several numerical experiments that our new model is able to obtain more interpretable results than standard NMF models. Moreover, we show that it is able to recover the sources even when the number of sources present into the mixed signal is overestimated. In fact, our model automatically sets sources to zero in this situation, hence performs model order selection automatically. 
\end{abstract}
\begin{IEEEkeywords}
nonnegative matrix factorization, 
$\beta$-divergences, 
minimum-volume regularization, 
identifiability, 
blind audio source separation, 
model order selection
\end{IEEEkeywords}


\IEEEpeerreviewmaketitle


\section{Introduction}

Blind audio source separation concerns the techniques used to extract unknown signals called sources from a mixed audio signal $x$. 
In this paper, we assume that the audio signal is recorded with a single microphone. Considering a mixed signal composed of various audio sources, the blind audio source separation consists in isolating and extracting each of the sources on the basis of the single recording. Usually, the only known information is the number of estimated sources present in the mixed signal. The blind source separation problem is said to be underdetermined as there are fewer sensors (only one in our case) than sources. It then appears necessary to find additional information to make the problem well posed. The most common technique used for this kind of problem is to get some form of redundancy in the mixed signal in order to make it overdetermined. This is typically done by computing the spectrogram which represents the signal in the time and frequency domains simultaneously (splitting the signals into overlapping time frames). The computation of spectrograms can be summarized as follows: short time segments are extracted from the signal and multiplied element wise by a window function or ``smoothing'' window of size $F$. Successive windows overlap by a fraction of their length, which is usually taken as 50\%. On each of these segments, a discrete Fourier transform is computed and stacked column-by-column in a matrix $X$. Thus, from a one-dimensional signal $x \in \mathbb{R}^{T}$, we obtain a complex matrix $X \in \mathbb{C}^{F \times N}$ called spectrogram where $F \times N \simeq 2T$ (due to the 50\% overlap between windows). Note that the length of the window determines the shape of the spectrogram. 
These preliminary operations correspond to computing the short time Fourier transform (STFT), which is given by the following formula: for $1 \leq f \leq F$ and $1 \leq n \leq N$, $X_{f,n}= \sum_{j=0}^{F-1} w_j x_{nL+j} e^{(-i \frac{2\pi f j }{F})}$,
 where $w \in \mathbb{R}^{F}$ is the smoothing window of size $F$, $L$ is a shift parameter (also called hop size), and $H=F-L$ is the overlap parameter. 
The number of rows corresponds to the frequency resolution. Letting $f_s$ be the sampling rate of the audio signal, consecutive rows correspond to frequency bands that are $f_s/F$ Hz apart. 

The time-frequency representation of a signal highlights two of its fundamental properties: 
sparsity 
and 
redundancy. 
Sparsity comes from the fact that most real signals are not active at all frequencies at all time points. 
Redundancy comes from the fact that frequency patterns of the sources repeat over time. 
Mathematically, this means that the spectrogram is a low-rank matrix. 
These two fundamental properties led sound source separation techniques to integrate algorithms such as nonnegative matrix factorization (NMF). Such techniques retrieve sensible solutions even for single-channel signals.  



\subsection{Mixing assumptions}\label{mixAss}

Given $K$ source signals $s^{(k)} \in \mathbb{R}^{T}$ for $1 \leq k \leq K$, we assume the acquisition process is well modelled by a linear instantaneous mixing model: 
\begin{equation} \label{eq:1}
 x(t)=\sum_{k=1}^{K}  \modif s^{(k)}(t)\endmodif  \quad \mathrm{with} \thinspace t=0,...,T-1\,.
\end{equation}
Therefore, for each time index $t$, the mixed signal $x(t)$ from a single microphone is the sum of the $K$ source signals. It is standard to assume that microphones are linear as long as the recorded signals are not too loud. If signals are too loud, they are usually clipped. The mixing process is modelled as instantaneous as opposed to convolutive used to take into account sound effects such as reverberation. The source separation problem consist in finding source estimates $\hat s^{(k)}$ of $s^{(k)}$ sources for all $k \in \{1,\dots,K\}$. 
Let us denote $S$ the linear STFT operator, and let $S^{\dagger}$ be its conjugate transpose. 
We have  $S^{\dagger}S=F I$, where $I$ is the identity matrix of appropriate dimension.  
For the remainder of this paper, $S^{\dagger}$  stands for the inverse short time Fourier transform. Note that the term inverse is not meant in a mathematical sense. Indeed the STFT is not a surjective transformation from $\mathbb{R}^{T}$ to $\mathbb{C}^{F \times N}$. In other words, each spectrogram or each matrix with complex entries is not necessarily the STFT of a real signal; see~\cite{Lefevre_phd} and~\cite{Magron_phd} for more details.  
By applying the STFT operator $S$ to~\eqref{eq:1}, 
we obtain the mixing model in the time-frequency domain :
\begin{equation*} 
 X 
 = S(x(t)) 
 = S \left( \sum_{k=1}^{K} \modif s^{(k)}(t\endmodif) \right) 
 = \sum_{k=1}^{K} S^{(k)}, 
\end{equation*} 
where $S^{(k)}$ is the STFT of the source $k$, that is, the spectrogram of source $k$. 
To identify the sources, we use in this paper the amplitude spectrogram $V = |X|  \in \mathbb{R}_{+}^{F \times N}$ defined as $V_{fn}=\left| X_{fn} \right| $ for all $f$, $n$. We assume that  
$V = \sum_{k=1}^K \left| S^{(k)} \right|$, 
which means that there is no sound cancellation between the sources, which is usually the case in most signals. 
Finally, we assume that the source spectrograms  $\left| S^{(k)} \right|$ are well approximated by nonnegative rank-one matrices. This leads to the NMF model described in the next section. Note that a source can be made of several rank-one factors in which case a post-processing step will have to recombine them a posteriori (e.g., looking at the correspondence in the activation of the sources over time). 
Note also that we focus on the NMF stage of the source separation which factorizes $V$ into the source spectrograms. 
For the phases reconstruction, which is a highly non-trivial problem, 
we consider a naive reconstruction procedure consisting in keeping the same phase as the input mixture for each source~\cite{Lefevre_phd}.

\subsection{NMF for audio source separation} \label{NMF}

Given a non-negative matrix $V \in \mathbb{R}_{+}^{F \times N}$ (the spectrogram) and a positive integer $K \ll \min(F,N)$ (the number of sources, called the factorization rank), 
NMF aims to compute two non-negative matrices $W$ with $K$ columns and $H$ with $K$ rows such that $V\approx WH$. 
NMF approximate each column of $V$ by a linear combination of the columns of $W$ weighted by the components of the corresponding column of $H$~\cite{algoNMFlee}. 
When the matrix $V$ corresponds to the amplitude spectrogram or the power spectrogram of an audio signal, we have that 

\noindent $\bullet$ $W$ is referred as the dictionary matrix and each column corresponds to the spectral content of a source,  and 

\noindent $\bullet$
  $H$ is the activation matrix specifying if a source is active at a certain time frame and in which intensity. 

In other words, each rank-one factor $W(:,k)H(k,:)$ will correspond to a source: the $k$th column $W(:,k)$ of $W$
 is the spectral content of source $k$, 
 and the $k$th row $H(k,:)$ \modif of \endmodif $H$  is its activation over time. 
  To compute $W$ and $H$, NMF requires to solve the following optimization problem 
\[ 
 \min_{W \geq 0,H \geq 0} 
D\left(V|WH\right) = \sum_{f,n} d(V_{fn}|[WH]_{fn}), 
\] 
where $A \geq 0$ means that $A$ is component-wise nonnegative, and 
$d(x|y)$ is an appropriate measure of fit. 
In audio source separation, a common measure of fit is the discrete $\beta$-divergence denoted $d_{\beta}(x|y)$ and equal to 
\[
\left\{
  \begin{array}{lr}
    \frac{1}{\beta \left(\beta-1\right)}   \left(x^{\beta}+\left(\beta-1\right)y^{\beta}-\beta xy^{\beta-1}\right)  \text{ for } 
     \beta \neq 0,1, \\
    x\log\frac{x}{y}-x+y           
         \text{ for } \beta=1,  \\
    \frac{x}{y}-\log\frac{x}{y}-1  
        \text{ for } \beta=0. 
  \end{array}
\right.
\] 
For $\beta=2$, this the standard squared Euclidean distance, that is, the squared Frobenius norm $||V-WH||_F^2$. 
For $\beta=1$ and $\beta=0$, 
the $\beta$-divergence corresponds to the Kullback-Leibler (KL) divergence and the Itakura-Saito (IS) divergence, respectively. 
The error measure which should be chosen accordingly with the noise statistic assumed on the data. 
The Frobenius norm assumes i.i.d.\@ Gaussian noise, KL divergence assumes additive Poisson noise, 
and the IS divergence assumes multiplicative Gamma noise~\cite{fevotte2009nonnegative}. 
The $\beta$-divergence $d_{\beta}(x|y)$ is homogeneous of degree $\beta$: $d_{\beta}(\lambda x|\lambda y)=\lambda^{\beta} d_{\beta}(x|y)$. It implies that factorizations obtained with $\beta > 0$ (such as the Euclidean distance or the KL divergence) will rely more heavily on the largest data values and less precision is to be expected in the estimation of the low-power components. The IS divergence ($\beta =0$) is scale-invariant that is $d_{IS}(\lambda x|\lambda y)=d_{IS}(x|y)$ \cite{Fevotte_betadiv}. The IS divergence is the only one in the $\beta$-divergences family to possess this property. It implies that time-frequency areas of low power are as important in the divergence computation as the areas of high power. 
  This property is interesting in audio source separation as low-power frequency bands can perceptually contribute as much as high-power frequency bands. Note that both KL and IS divergences are more adapted to audio source separation than Euclidean distance as it is built on logarithmic scale as human perception; see \cite{Lefevre_phd} and \cite{Fevotte_betadiv}.
  Moreover, the $\beta$-divergence is only convex with respect to $W$ (or $H$) if $\beta \geq 1$. 
Otherwise, the objective function is non-convex. This implies that, for $\beta <1$, even the problem of inferring $H$ with $W$ fixed is non-convex. 
For more details on $\beta$-divergences; 
see~\cite{Fevotte_betadiv}.

\subsection{Contribution and outline of the paper} \label{sec:contr}

In Section~\ref{sec:model}, we propose a new NMF model, referred to as minimum-volume $\beta$-NMF (min-vol $\beta$-NMF), to tackle the audio source separation problem. This model penalizes the columns of the dictionary matrix $W$ so that their convex hull has a small volume. 
To the best of our knowledge, this model is novel in two aspects: 
(1)~it is the first time a minimum-volume penalty is associated with a $\beta$-divergence for $\beta \neq 2$ and it is the first time such models are used in the context of audio source separation, 
and 
(2)~as opposed to most previously proposed minimum-volume NMF models, our model imposes a normalization constraints on the factor $W$ instead of $H$. As far as we know, the only other paper that used a normalization of $W$ is~\cite{zhou2011minimum} but the authors did not justify this choice compared to the normalization of $H$ (the choice seems arbitrary, motivated by the `elimination of the norm indeterminacy'), nor provided theoretical guarantees. 
In this paper, we explain why normalization of $W$ is a better choice in practice, and we prove that, under some mild assumptions and in the noiseless case, this model provably identify the sources; see Theorem~\ref{mainth}. To the best of our knowledge, this is the first result of this type in the audio source separation literature.   
In Section~\ref{sec:algo}, we propose an algorithm to tackle min-vol $\beta$-NMF, focusing on the KL and IS divergences. 
The algorithm is based on multiplicative updates (MU) that are derived using the standard majorization-minimization framework, and that monotonically decrease the objective function. 
In Section~\ref{sec:numexp}, we present several numerical experiments, comparing min-vol $\beta$-NMF with standard NMF and sparse NMF. 
The two mains conclusions are that 
(1)~minimum-volume $\beta$-NMF performs consistently better to  identify the sources, 
and 
(2)~as opposed to NMF and sparse NMF, min-vol $\beta$-NMF is able to detect when the factorization rank is overestimated by automatically setting sources to zero.


\section{Minimum-volume NMF with $\beta$-divergences} \label{sec:model}

In this section, we present a new model of separation based on the minimization of $\beta$-divergences including a penalty term promoting solutions with minimum volume spanned by the columns of the dictionary matrix $W$.  
Section~\ref{sec:model:geo} recalls the geometric interpretation of NMF which motivated the use of a minimum volume penalty on the dictionary $W$.  
Section~\ref{sec:model:norma} discusses the new proposed normalization compared to previous minimum volume NMF models, and proves that min-vol $\beta$-NMF provably recovers the true factors $(W,H)$ under mild conditions and in the noiseless case; see Theorem~\ref{mainth}.

\subsection{Geometry and the min-vol $\beta$-NMF model} \label{sec:model:geo}

As mentioned earlier, $V = WH$ means that each column of $V$ is a linear combination of the columns of $W$ weighted by the components of the corresponding column of $H$; in fact, 
$v_n = Wh_n$ for $n=1,...,N$, 
where $v_{n}$ denotes the $n$th column of data matrix $V$.  This gives to NMF a nice geometric interpretation: for all $n$ 
 \[
 v_{n} \in \cone ( W ) = \left\lbrace v \in \mathbb{R}^{F} | v=W \theta,  \theta \geq 0 \right\rbrace , 
\] 
meaning that the columns of $V$ are contained in the convex cone generated by the columns of $W$; see Figure~\ref{fig:cone} for an illustration.  
%
\begin{figure}[h]
\centering
\includegraphics[width=0.8\linewidth]{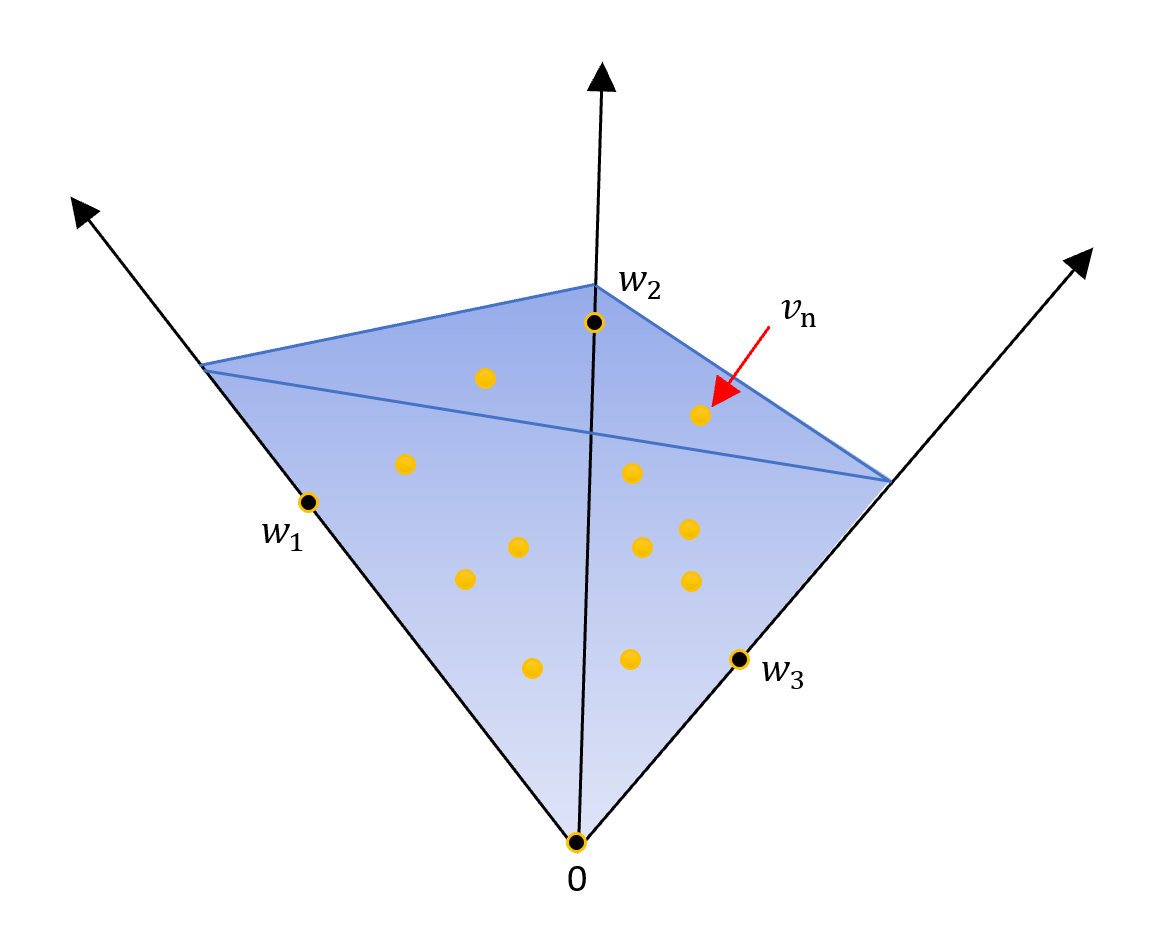}
\caption{Geometric interpretation of NMF for $K=3$~\cite{nmfidentifiable}. \label{fig:cone}} 
\label{fig:conichull}
\end{figure}
From this interpretation, it follows that, in general, NMF decompositions are not unique because there exists several (often, infinitely many) sets of columns of $W$ that span the convex cone generated by the data points; see for example~\cite{huang} for more details. 
Hence, NMF is in most cases ill-posed because the optimal solution is not unique. In order to make the solution unique (up to permutation and scaling of the columns of $W$ and the rows of $H$) hence making the problem well-posed and the parameters $(W,H)$ of the problem identifiable, a key idea is to look for a solution $W$ with minimum volume. Intuitively, we will look for the cone $\cone(W)$ containing the data points and as close as possible to these data points. 
The use of minimum-volume NMF has lead to a new class of NMF methods that outperforms existing ones in many applications such as document analysis and blind hyperspectral unmixing; see the recent survey~\cite{fu2019nonnegative}. 
Note that minimum-volume NMF implicitly enhances the factor $H$ to be sparse: the fact that $W$ has a small volume implies that many data points will be located on the facets of the $\cone(W)$ hence $H$ will be sparse.

Hence, in this paper, we consider the following model, referred to as min-vol $\beta$-NMF: 
\begin{equation}\label{eq:4}
\begin{aligned}
& \underset{W(:,j) \in \Delta^{F} \forall j, H\geq 0}{\text{min}}
& D_{\beta}(V|WH) + \lambda \text{vol}(W), 
\end{aligned}
\end{equation} 
where 
$\Delta ^{F}=\left\lbrace x \in \mathbb{R}^{F}_{+} \big| \sum_{i=1}^F x_{i}=1 \right\rbrace$ is the unit simplex, 
$\lambda$ is a penalty parameter 
and 
$\text{vol}(W)$ is a function that measures the volume spanned by the columns of $W$.  
In this paper, we use 
$\text{vol}(W)=\logdet(W^{T}W+\delta I)$, 
 where $\delta$ is a small positive constant that prevents $\logdet(W^{T}W)$ to go to $-\infty$ when $W$ tends to a rank-deficient matrix (that is, when $r=\text{rank}(W)< K$). 
The reason for using such a measure is that $\sqrt{\text{det}\left(W^{T}W \right)}/K!$ is the volume of the convex hull of the columns of $W$ and the origin. This measure is one of the most widely used ones, and has been shown to perform very well in practice~\cite{robustvol, ang2019algorithms}. 
Moreover, the criterion $\logdet(W^TW + \delta I)$ is able to distinguish two rank-deficient solutions and favour solutions for $W$ with smaller volume~\cite{volminleplat}. Finally, as we will illustrate in Section~\ref{sec:numexp}, this criterion is able to identify the right number of source even when $K$ is overestimated, by putting some rank-one factors to zero.

\subsection{Normalization and identifiability of min-vol $\beta$-NMF}\label{sec:model:norma}

As mentioned above, under some appropriate conditions on $V=WH$, minimum-volume NMF models will provably recover the ground-truth $(W,H)$ that generated $V$, up to permutation and scaling of the rank-one factors. 
The first identifiability results for minimum-volume NMF models assumed that the entries in each column of $H$ sum to one, that is, that $H^Te=e$ where $e$ is the all-one column vector whose \modif dimension \endmodif is clear from the context, meaning \modif that \endmodif $H$ is column stochastic~\cite{huang, lin2015identifiability}. 
Under this condition, each column of $V$  lies in the \textit{convex hull} of the columns of $W$; see Figure~\ref{fig:convexhull} for an illustration.   
\begin{figure}[H]
\centering
\includegraphics[width=0.8\linewidth]{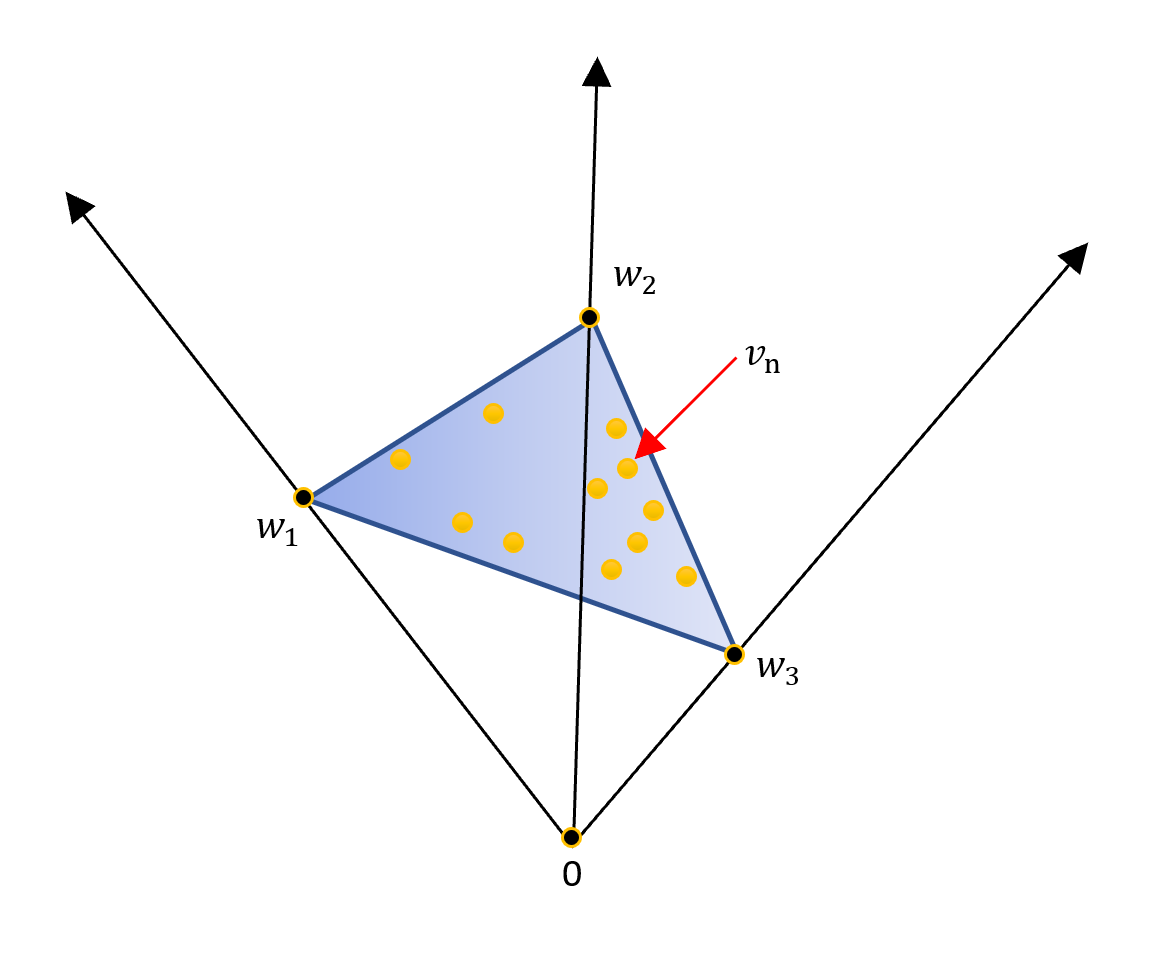}
\caption{Geometric interpretation of NMF when $K=3$ and $H$ is column stochastic~\cite{nmfidentifiable}.}
\label{fig:convexhull}
\end{figure}  
Under the three assumptions that 
(1)~$H$ is column stochastic, 
(2)~$W$ is full column rank, and 
(3)~$H$ satisfies the sufficiently scattered condition, 
 minimizing the volume of $\conv(W)$ such that $V=WH$ recover the true underlying factors, up to permutation and scaling. Intuitively, the sufficiently scattered condition requires $H$ to be sparse enough so that data points are located on the facets of $\conv(W)$; 
see Appendix~\ref{app:iden} for a formal definition. 
 The sufficiently scattered condition makes sense for most audio source data sets as it is reasonable to assume that, for most time points, only a few sources are active hence $H$ is sparse; see~\cite{fu2019nonnegative} for more details on the sufficiently scattered condition. 
Note that the sufficiently scattered condition is a generalization of the separability condition which requires $W=V\left(:, \mathcal{J} \right)$ for some index set $\mathcal{J}$ of size $K$~\cite{arora2016computing}.   
However, separability is a much stronger assumption as it requires that, for each sources, there exists a time point where only that source is active. 
Note that although min-vol NMF guarantees identifiability, the corresponding optimization problem~\eqref{eq:4} is still hard to solve in general, as for the original NMF problem~\cite{NMF_complexity}.

Despite this nice result, the constraint~$H^Te=e$ makes the NMF model less general and does not apply to all data sets. 
In the case where the data does not naturally belong to a convex hull, one needs to normalize the data points so that their entries sum to one so that $H^Te=e$ can be assumed without loss of generality (in the noiseless case). 
This normalization can sometimes increase the noise and might greatly influence the solution, hence are usually not recommended in practice; see the discussion in~\cite{fu2019nonnegative}. 

In~\cite{nmfidentifiable}, authors show that  identifiability still holds when the condition that $H$ is column stochastic is relaxed to $H$ being row stochastic. 
As opposed to column stochasticity, 
row stochasticity of $H$ can be assumed without loss of generality since any factorization $WH$ can be properly normalized so that this assumption holds. In fact, 
$WH = \sum_{k=1}^K (a_k W(:,k)) (H(k,:)/a_k)$ for any $a_k > 0$ for $k=1,\dots,K$.
In other terms, letting $A$ be the diagonal matrix with $A(k,k) = a_k = \sum_{j=1}^n H(k,j)$ for $k=1,\dots,K$, we have $WH = (WA)(A^{-1}H) = W'H'$ where $H'=A^{-1}H$ is row stochastic.  

Similarly as done in~\cite{nmfidentifiable}, we prove in this paper that requiring $W$ to be column stochastic (which can also be made without loss of generality) also leads to  identifiability. 
Geometrically, the columns of $W$ are constrained to be on the unit simplex. Minimizing the volume still makes a lot of sense: we want the columns of $W$ to be as close as possible to one another within the unit simplex. In Appendix~\ref{app:iden}, we prove the following theorem. 
\begin{theorem} \label{mainth} 
Assume $V=W^{\#}H^{\#}$ with \mbox{$\text{rank}(V) = K$}, 
$W^{\#} \geq 0$ and $H^{\#}$ satisfies the sufficiently scattered condition (Definition~\ref{def4} in Appendix~\ref{app:iden}).   
Then the optimal solution of 
\begin{align}\label{eq:8b}  
\min_{W\in \mathbb{R}^{F \times K},H\in \mathbb{R}^{K \times N}} & \logdet\left( W^{T}W \right)  \\
 \text{such that} & \quad  V=WH, W^T e = e, H \geq 0, \nonumber    
\end{align}
recovers $(W^{\#},H^{\#})$ up to permutation and scaling. 
\end{theorem}
\begin{proof}
See Appendix~\ref{app:iden}.  
\end{proof}


In noiseless conditions, 
replacing $W^Te=e$ with $He=e$ in~\eqref{eq:8b} leads to the same identifiability result; see \cite[Theorem 1]{nmfidentifiable}. 
Therefore, in noiseless conditions and under the conditions of Theorem~\ref{mainth}, both models return the same solution up to permutation and scaling. 
However, in the presence of noise, we have observed that the two models may behave very differently.
In fact, we advocate that the constraint $W^{T}e=e$ is better suited for noisy real-world problems, which we have observed on many numerical examples. In fact, we have observed that the normalization $W^{T}e=e$ is much less sensitive to noise and returns much better solutions.  
The reason is mostly twofold: 

\noindent (i) As described above, using the normalization $He=e$ amounts to multiply $W$ by a diagonal matrix whose entries are the $\ell_1$ norms of the rows of $H$. 
Therefore, the columns of $W$ that correspond to dominating (resp.\@ dominated) sources, that is, sources with much more (resp.\@ less) power and/or active at many (resp.\@ few) time points, will have much higher (resp.\@ lower) norm. 
Therefore, the term $\logdet(W^{T}W+\delta I)$ is much more influenced by the dominating sources and will have difficulties to penalize the dominated sources.  
In other terms, the use of the term $\logdet(W^{T}W+\delta I)$ with the normalization $He=e$ implicitly requires that the rank-one factors $W(:,k)H(k,:)$ for $k=1,\dots,K$ are well balanced, that is, have similar norms. This is not the case for many real (audio) signals. 

\noindent (ii) As it will be explained in Section~\ref{sec:algo}, the update of $W$ needs the computation of the matrix $Y$ which is the inverse of ${W}^{T}{W}+\delta I$--this terms appears in the gradient with respect to $W$ of the objective function. 
The numerical stability for such operations is related to the condition number of $W^{T}W+\delta I$. For a $\ell_{1}$ normalization on the columns of $W$, the condition  number is bounded above as follows:  
$\cond( W^TW + \delta I ) 
 = \frac{ \sigma_{\max}(W^TW+\delta I) }{ \sigma_{\min}(W^TW+\delta I) }  
 = \frac{ \sigma_{\max}(W)^2 + \delta }{ \sigma_{\min}(W)^2 + \delta } 
 \leq \frac{ \left(\sqrt{K} \max_k ||W(:,k)||_2\right)^2 + \delta }{\delta} 
\leq  1 + \frac{K}{\delta}$,  
where $\sigma_{\min}(W)$ and $\sigma_{\max}(W)$ are the smallest and largest singular values of $W$, respectively. 
In the numerical experiments, we use $\delta = 1$.  
On the other hand, the normalization $He=e$ may lead to arbitrarily large values for the condition number of $W^TW+\delta I$, which we have observed numerically on several examples. This issue can be mitigated with the use of the normalization $He=\rho e$ for some $\rho > 0$ sufficiently large for which identifiabilty still holds~\cite{nmfidentifiable}. However, it still performs worse because of the first reason explained above. 

For these reasons, we believe that the model~\eqref{eq:8b} would also be better suited (compared to the normalization on $H$) in other contexts; for example for document classification~\cite{fu2018anchor}.


\section{Algorithm for min-vol $\beta$-NMF} \label{sec:algo}

Most NMF algorithms alternatively update $H$ for $W$ fixed and vice versa, and we adopt this strategy in this paper. 
For $W$ fixed, \eqref{eq:4} is equivalent to standard NMF and we will use the MU that have already been derived in the literature~\cite{algoNMFlee, Fevotte_betadiv}. 

To tackle \eqref{eq:4} for $H$ fixed, let us consider  
\begin{equation}\label{eq:10}
\begin{aligned}
& \underset{W \geq 0}{\text{min}}
& & F(W)=D_{\beta}(V|WH)+ \lambda\logdet(W^{T}W+\delta I). 
\end{aligned}
\end{equation} 
Note that, for now, we have discarded the normalization on the columns of $W$. 
In our algorithm, we will use the update for $W$ obtained by solving~\eqref{eq:10} as a descent direction along with a line search procedure to integrate the constraint on $W$. This will ensure that the objective function $F$ is non-increasing at each iteration.  
In the following sections we derive MU for $W$ that decrease the objective in~\eqref{eq:10}. We follow the standard majorization-minimization framework~\cite{sun2017majorization}. 
First, an auxiliary function, which we denote $\bar{F}$, is constructed so that it majorizes the objective. 
An auxiliary function for $F$ at point $\tilde{W}$ is defined as follows.  
\begin{definition}\label{def2}
The function $\bar{F}(W|\tilde{W}): \Omega \times \Omega \rightarrow \mathbb{R}$ is an auxiliary function for $F\left( W \right):\Omega \rightarrow \mathbb{R}$ at $\tilde{W} \in \Omega$ if the conditions  
$\bar{F}(W|\tilde{W}) \geq F\left( W \right)$ 
for all $W \in \Omega$ and $\bar{F}(\tilde{W}|\tilde{W}) = F( \tilde{W} )$ 
are satisfied.
\end{definition}
Then, 
the optimization of $F$ can be replaced by an iterative process that minimizes $\bar{F}$. More precisely, the new iterate $W^{(i+1)}$ is computed by minimizing exactly the auxiliary function at the previous iterate $W^{(i)}$. This guarantees $F$ to decrease at each iteration. 
\begin{lemma}\label{lemmamono}
Let $W, W^{(i)} \geq 0$, and let $\bar{F}$ be an auxiliary function for $F$ at $W^{(i)}$. 
Then $F$ is non-increasing under the update
$W^{(i+1)} = \underset{W \geq 0}{\text{argmin}} \bar{F}(W|W^{(i)})$. 
\end{lemma}
\begin{proof}
In fact, we have by definition that 
$F(W^{(i)}) = \bar{F}(W^{(i)}|W^{(i)}) 
\geq 
\underset{W}{\text{min}} \bar{F}(W|W^{(i)}) 
= \bar{F}(W^{(i+1)}|W^{(i)})\geq F(W^{(i+1)})$. 
\end{proof}

The most difficult part in using the majorization-minimization framework is to design an auxiliary function that is easy to optimize. Usually such auxiliary functions are separable (that is, there is no interaction between the variables so that each entry of $W$ can be updated independently) and convex.

\subsection{Separable auxiliary functions for $\beta$-divergences}\label{section31}

For the sake of completeness, we briefly recall the auxiliary function proposed in~\cite{Fevotte_betadiv} for the data fitting term. It consists in majorizing the convex part of the $\beta$-divergence using Jensen's inequality and majorizing the concave part by its tangent (first-order Taylor approximation). We have 
  \begin{equation}\label{eq:3}
  d_{\beta}(x|y)= \check{d}_{\beta}(x|y)+\hat{d}_{\beta}(x|y)+\bar{d}_{\beta}(x|y), 
  \end{equation}
where $\check{d}$ is convex function of $y$, $\hat{d}$ is a concave function of $y$ and $\bar{d}$ is a constant of $y$; see Table~\ref{table:conv_concav_decomp}.  
\begin{center}
\begin{table}[h!]
\begin{center}
\caption{Differentiable convex-concave-constant decomposition of the $\beta$-divergence under the form \eqref{eq:3}~\cite{Fevotte_betadiv}. }
\label{table:conv_concav_decomp}
\begin{tabular}{|c|c|c|c|}
\hline 
      & $\check{d}(x|y)$  &  $\hat{d}(x|y)$ & $\bar{d}(x)$   \\  \hline 
 $\beta=0$      
 & $xy^{-1}$     & $\log(y)$  & $x(\log(x)-1)$  \\
 $\beta \in [1,2] $ 
 & $d_{\beta}(x|y)$ & 0 & 0    \\ \hline  
\end{tabular} 
\end{center}
\end{table}
\end{center}

The function $D_{\beta}(V|WH)$ can be written as $\sum_{f} D_{\beta}(v_{f}|w_{f}H)$ where $v_{f}$ and $w_{f}$ are respectively the $f$th row of $V$ and $W$. 
Therefore we only consider the optimization over one specific row of \modif $W$\endmodif. To simplify notation, we denote iterates $w^{(i+1)}$ (next iterate) and $w^{(i)}$ (current iterate) as $w$ and $\tilde{w}$, respectively. 

\begin{lemma}[\cite{Fevotte_betadiv}] \label{defG}
Let $\tilde{v} = \tilde{w}H$ \modif and $\tilde{w}$ be such that $\tilde{v_{n}} >0$ for all $n$ and $\tilde{w_{k}} >0$ for all $k$. \endmodif  
Then the function 
\begin{equation}\label{eq:20}
\begin{aligned}
G(w|\tilde{w})&=\sum_{n}\left[\sum_{k} \frac{\tilde{w_{k}}h_{kn}}{\tilde{v_{n}}}\check{d}(v_{n}|\tilde{v_{n}}\frac{w_{k}}{\tilde{w_{k}}})\right] +\bar{d}(v_{n})  \\
&+\left[\hat{d}^{'}(v_{n}|\tilde{v_{n}})\sum_{k}(w_{k}-\tilde{w_{k}})h_{kn}+\hat{d}(v_{n}|\tilde{v_{n}}) \right]
\end{aligned}
\end{equation} 
is an auxiliary function for $\sum_{n}d(v_{n}|\left[wH\right]_{n})$ at $\tilde{w}$.
\end{lemma}

\subsection{A separable auxiliary function for the minimum-volume regularizer}\label{section32}

The minimum-volume regularizer $\logdet(W^{T}W+\delta I)$ is a non-convex function. 
However, it can be upper-bounded using the fact that $\logdet(.)$ is a concave function so that its first-order Taylor approximation provides an upper bound; see for example~\cite{robustvol}. For any positive-definite matrices $A$ and $B \in \mathbb{R}^{K \times K}$, 
we have:
\begin{align*}
 \logdet\left( B \right) & \leq \logdet\left( A \right)+ \tr\left( A^{-1} \left( B-A \right) \right) \\ 
& = \tr\left( A^{-1} B \right) + \logdet\left( A \right)- K \,.
\end{align*}
This implies that for any $W, Z \in \mathbb{R}^{F \times K}$, we have 
\begin{equation} \label{eq:29}
\begin{aligned}
& \logdet(W^{T}W+\delta I) \leq l(W,Z),   
\end{aligned}
\end{equation}
where $l(W,Z) = \tr\left( Y W^{T}W \right) + \logdet\left( Y^{-1} \right)- K$, $Y=(Z^{T}Z+\delta I)^{-1}$ with $\delta >0$. 
Note that $Z^{T}Z+\delta I$ is positive definite hence is invertible and its inverse $Y$ is also positive definite.
Finally $l(W,Z)$ is an auxiliary function for $\logdet(W^{T}W+\delta I)$ at $Z$. 
However, it is quadratic and not separable hence non-trivial to optimize over the nonnegative orthant.  
The non-constant part of $l(W,Z)$ can be written as $ \sum_{f}w_{f} Y w_{f}^T$ 
where $w_{f}$ is the $f$th row of $W$. 
Henceforth we will focus on one particular row vector $w$ with $l\left( w \right) =  w^{T} Y w$ which will be further considered as a column vector of size $K \times 1$. 
\begin{lemma}\label{lem1quad}
Let $w, \tilde{w} \in \mathbb{R}^K_+$ \modif be such that $\tilde{w_{k}} >0$ 
for all~$k$\endmodif, 
$Y=Y^{+}-Y^{-}$ 
with $Y^{+}=\max\left(Y,0 \right)$ 
and $Y^{-}=\max\left(-Y,0 \right)$,  and 
\modif $\Phi\left( \tilde{w} \right)$ \endmodif be the diagonal matrix 
\modif $\Phi\left( \tilde{w} \right) = 
\diag \left(  2
 \frac{[Y^{+}\tilde{w}+Y^{-}\tilde{w}]}{[\tilde{w}]} \right)$  where $\frac{\left[ A \right]}{\left[ B \right]}$ is the component-wise division between $A$ and $B$ \endmodif, and $\Delta w = w-\tilde{w}$. 
Then 
\begin{equation} \label{auxftcl}
\bar{l}(w|\tilde{w}) = l(\tilde{w}) +  \Delta w^{T}  \nabla l \left( \tilde{w}\right) + 
\frac{1}{2} \Delta w^{T} \modif \Phi(\tilde{w}) \endmodif \Delta w,  
\end{equation} 
is a separable auxiliary function for $l\left( w \right)$=$w^{T}Yw$ at $\tilde{w}$. 
\end{lemma} 
\begin{proof}
See Appendix~\ref{app:lem3}.  
\end{proof}

\modif
\begin{remark}[Choice of the auxiliary function]
A simpler choice for the auxiliary function would be  to replace  
$\Phi(\tilde{w})$ with $2\lambda_{\max}(Y) I$ 
where $\lambda_{\max}(Y)$ is the largest eigenvalue of $Y$ (the constant $2$ appears  because $l\left( w \right)=  w^{T}Yw$ while there is a factor $1/2$ in front of $\Phi(\tilde{w})$). 
However, it would lead to a worse approximation. In particular if $Y$ is a diagonal matrix (since $Y \succ 0$, these diagonal elements are positive), our choice gives $\Phi(\tilde{w}) = 2Y$ for any $\tilde{w} > 0$, meaning that the auxiliary function matches perfectly the function $l\left( w \right)$.  
This would not be the case for the choice $2\lambda_{\max}(Y) I$ (unless $Y$ is a scaling of the identity matrix).  
\end{remark}
\endmodif

\subsection{Auxiliary function for min-vol $\beta$-NMF}  \label{section33}

Based on the auxiliary functions presented in Sections \ref{section31} and \ref{section32}, we can directly derive a  separable auxiliary function $\bar{F}(W|\tilde{W})$ for min-vol $\beta$-NMF~\eqref{eq:4}. 
\begin{corollary}\label{corro1}
For $W,H\geq 0$, $\lambda >0$, $Y=(\tilde{W}^{T}\tilde{W}+\delta I)^{-1}$ with $\delta > 0$ and the constant $c=\logdet\left( Y^{-1} \right)+ K$, the function 
\begin{equation*}\label{eq:40}
\begin{aligned}
& \bar{F}(W|\tilde{W}) 
= \sum_{f} G\left( w_f |\tilde{w}_f \right) + \lambda \left( \sum_{f} \bar{l}\left( w_f |\tilde{w}_f\right) + c\right), 
\end{aligned}
\end{equation*}
where $G$ is given by \eqref{eq:20} and 
$\bar{l}$ by \eqref{auxftcl}, 
is a convex and separable auxiliary function for  
$F(W)=D_{\beta}(V|WH) + \lambda \logdet(W^{T}W+\delta I)$ at $\tilde{W}$. 
\end{corollary}
\begin{proof} 
This follows directly from 
Lemma~\ref{defG}, 
Equation~\eqref{eq:29} and 
Lemma~\ref{lem1quad}.
\end{proof}


In the following section, we provide explicitly MU for the KL divergence ($\beta =1$) 
by finding a closed-form solution for the minimization of $\bar{F}$. 
In Appendix~\ref{app:is}, we provide the MU for the IS divergence ($\beta = 0$). 
Due to the lack of space, the other cases are not treated explicitly but can be in a similar way.  
For the same reason, we will only compare KL NMF models in the numerical experiments (Section~\ref{sec:numexp}).

\subsection{Algorithm for min-vol KL-NMF}

As before, let us focus on a single row of $W$, denoted $w$, as the \modif objective \endmodif function $F(W)$ is separable by row.     
For $\beta=1$, 
the derivative of the auxiliary function $\bar{F}(w|\tilde{w})$ with respect to a specific coefficient $w_{k}$ is given by 
$\nabla_{w_{k}} \bar{F}(w|\tilde{w}) = \sum_{n}h_{kn}-\sum_{n}h_{kn}\frac{\tilde{w}_{k} v_{n}}{w_{k}\tilde{v}_{n}}+2\lambda \left[ Y \tilde{w} \right]_{k} + 2 \lambda \left[ \diag \left( \frac{Y^{+}\tilde{w}+Y^{-}\tilde{w}}{\tilde{w}} \right) \right]_{k} w_{k} 
- 2 \lambda \left[ \diag \left( \frac{Y^{+}\tilde{w}+Y^{-}\tilde{w}}{\tilde{w}} \right) \right]_{k} \tilde{w}_{k}$.  
Due to the separability, we set the derivative to zero to obtain the closed-form solution, which is given in Table~\ref{table:KL-MU} in matrix form. 
\begin{center}
\begin{table*}[h!]
\begin{center}
\caption{Multiplicative update for min-vol KL-NMF. }
\label{table:KL-MU}
\begin{tabular}{l}
$\hspace{2cm}  W=\tilde{W} \odot  \frac{\left[\left[ \left[ J_{F,N}  H^{T} - 4\lambda \left( \tilde{W} Y^{-}\right)\right]^{.2}+8 \lambda \tilde{W} \left( Y^{+}+Y^{-}\right)\odot \left( \frac{\left[ V \right]}{\left[\tilde{W}H\right]} H^{T} \right)\right]^{.\frac{1}{2}}- 
\left( J_{F,N}  H^{T} - 4\lambda \left( \tilde{W} Y^{-}\right) \right) 
\right]}
{\left[ 4\lambda \tilde{W} \left( Y^{+}+Y^{-}\right)\right] }$,  \\  
where $A \odot B$ 
(resp.\@ $\frac{\left[ A \right]}{\left[ B \right]}$) is the Hadamard product 
(resp.\@ division) between $A$ and $B$, $A^{(.\alpha)}$ is the element-wise $\alpha$ exponent of $A$, \\ 
$J_{F,N}$ is the $F$-by-$N$ all-one matrix, 
and $Y=Y^{+}-Y^{-}=(\tilde{W}^{T}\tilde{W}+\delta I)^{-1}$ with $\delta > 0$, $Y^{+} \geq, Y^{-} \geq 0$, $\lambda > 0$.  
\end{tabular} 
\end{center}
\end{table*}
\end{center}
\modif 
Note that although the closed-form solution has a negative term in the numerator of the multiplicative factor (see Table~\ref{table:KL-MU}), they always remain nonnegative given that $V, H$ and $\tilde{W}$ are nonnegative.   
In fact, the term before the minus sign is always larger than the  term after the minus sign: $J_{F,N}  H^{T} - 4\lambda ( \tilde{W} Y^{-})$ is squared (component wise) and added a positive term, hence the component-wise square root of that result is larger than $J_{F,N}  H^{T} - 4\lambda ( \tilde{W} Y^{-})$. 
 \endmodif

Algorithm~\ref{KLminvol} summarizes our algorithm to tackle~\eqref{eq:4} for $\beta=1$ which we refer to as min-vol KL-NMF. 
Note that the update of $H$ (step 4) is the one from~\cite{algoNMFlee}. 
More importantly, note that we have incorporated a line-search for the update of $W$. In fact, although the MU for $W$ are guaranteed to decrease the objective function, they do not guarantee that $W$ remains normalized, that is, that $||W(:,k)||_1 = 1$ for all $k$. Hence, we normalize $W$ after it is updated (step 10), and we normalize $H$ accordingly so that $WH$ remains unchanged. 
When this normalization is performed, the 
$\beta$-divergence part of $F$ is unchanged but the minimum-volume penalty will change so that the objective function $F$ might increase.  
In order to guarantee non-increasingness, 
we integrate a simple backtracking line-search procedure; see steps 11-16 of Algorithm~\ref{KLminvol}. \modif 
In summary, our MU provide a descent direction that preserved nonnegativity of the iterates, and we use a projection and a simple backtracking line-search to guarantee the monotonicity of the objective function, as in standard projected gradient descent methods. 
\endmodif

\algsetup{indent=2em}
\begin{algorithm}[ht!]
\caption{min-vol KL-NMF \label{KLminvol}}
\begin{algorithmic}[1] 
\REQUIRE A matrix $V \in \mathbb{R}^{M \times T}$, an initialization $H \in \mathbb{R}^{K \times T}_+$, an initialization $W  \in \mathbb{R}^{M \times K}$ , a factorization rank $K$, a maximum number of iterations maxiter, min-vol weight $\lambda>0$ and $\delta>0$
\ENSURE A rank-$K$ NMF $(W,H)$ of $V \approx WH$ with $W \geq 0$ and $H \geq 0$ . 
    \medskip  
\STATE $\gamma=1, Y=\left(W^{T}W+\delta I \right)^{-1}$
\FOR{$i$ = 1 : maxiter}
	\STATE \emph{\% Update of matrix $H$}
	\STATE  $H \leftarrow H \odot \frac{\left[  W^{T}\left( \frac{\left[ V\right]}{\left[ WH\right]} \right) \right]}{\left[  W^{T} J_{F,N} \right]}$ \quad 
	\STATE \emph{\% Update of matrix $W$}
	\STATE $Y \leftarrow \left(W^{T}W+\delta I \right)^{-1}$
	\STATE $Y^{+} \leftarrow \text{max}\left(Y,0\right)$
	\STATE $Y^{-} \leftarrow \text{max}\left(-Y,0\right)$
	\STATE $W^{+}$ is updated according to Table~\ref{table:KL-MU} 
	\STATE 
	$(W^{+}_{\gamma},H_{\gamma}) 
	= \text{normalize}\left( W^{+} , H \right)$
	\STATE \emph{\% Line-search procedure}
	\WHILE{$F\left( W^{+}_{\gamma},H_{\gamma}\right) > F\left( W,H\right) $}
		\STATE $\gamma \leftarrow \gamma \times 0.8$
\STATE $W^{+}_{\gamma} \leftarrow \left( 1-\gamma \right)W + \gamma W^{+}$ 
		\STATE $(W^{+}_{\gamma},H_{\gamma}) \leftarrow \text{normalize}\left(W^{+}_{\gamma}  , H\right)$
	\ENDWHILE
	\STATE $(W,H) \leftarrow (W^{+}_{\gamma},H_{\gamma})$ 
	\STATE \emph{\% Update of $\gamma$ \modif to avoid a vanishing stepsize \endmodif}
	\STATE $\gamma \leftarrow \text{min}\left(1,\gamma \times 1.2\right)$
\ENDFOR
\end{algorithmic}  
\end{algorithm}

It can be verified that the computational complexity of the min-vol KL-NMF is asymptotically equivalent to the standard MU for $\beta$-NMF, that is, it requires $\mathcal{O}\left(FNK \right)$ operations per iteration.  
\modif 
Indeed, all the main operations include matrix products with a complexity of $\mathcal{O}\left(FNK \right)$ and element-wise operations on matrices of size $F \times K$ or $K \times N$. 
Note that the inversion of the $K$-by-$K$ matrix $(W^T W + \delta I)$ requires $\mathcal{O}\left(K^3 \right)$ operations which is dominated by $\mathcal{O}\left(FNK \right)$ since $K \leq \min(F,N)$ (in fact, typically $K \ll \min(F,N)$ hence this term is negligible).  
Therefore, although Algorithm~\ref{KLminvol} will be slower than the baseline KL-NMF (that is, the standard MU) because of the additional terms to be computed and the line-search, the asymptotical computational cost is the same; see Table~\ref{table:runtimeperf} for runtime comparison. 
\endmodif

\section{Numerical experiments} \label{sec:numexp}

In this section we report an experimental comparative study of baseline KL-NMF, min-vol KL-NMF (Algorithm~\ref{KLminvol}) and sparse KL-NMF~\cite{Leroux} applied to the spectrogram of two monophonic piano sequences and a synthetic mix of a bass and drums. For the two monophonic piano sequences, the audio signals are true life signals with standard quality.  
Note that the sequences are made of pure piano notes, the number $K$ should therefore correspond to the number of notes present into the mixed signals. 
The comparative study is performed for several values of $K$ with a focus on the case where the factorization rank $K$ is overestimated. For all simulations, random initializations are used for $W$ and $H$, the best results among 5 runs are kept for the comparative study. 
In all cases, we use a Hamming window of size $F$=1024, and 50\% overlap between two frames. 
Sparse KL-NMF has a similar structure as min-vol KL-NMF, with a penalty parameter for the sparsity enforcing regularization. 
To tune these two parameters, we have used the same strategy for both methods: 
we manually tried a wide range of values and report the best results. 
The code is available from~\url{bit.ly/minvolKLNMF} (code written in MATLAB R2017a), and can be used to rerun directly all experiments below. They were run on a laptop computer with Intel Core
i7-7500U CPU $@$ 2.70GHz 4 and 32GB memory. 

\paragraph{Mary had a little lamb}\label{paraMary} The first audio sample is the first measure of ``Mary had a little lamb". The sequence is composed of three notes; $E_{4}$, $D_{4}$ and $C_{4}$, played all at once. The recorded signal is 4.7 seconds long and downsampled to $f_{s}=16000$Hz yielding $T$=75200 samples. STFT of the input signal $x$ yields a temporal resolution of 16ms and a frequency resolution of 31.25Hz, so that the amplitude spectrogram $V$ has $N$=294 frames and $F$=257 frequency bins. The musical score is shown on Figure~\ref{fig:mary_representations}. 
\begin{figure}[h]
\centering 
        \includegraphics[width=0.35\textwidth]{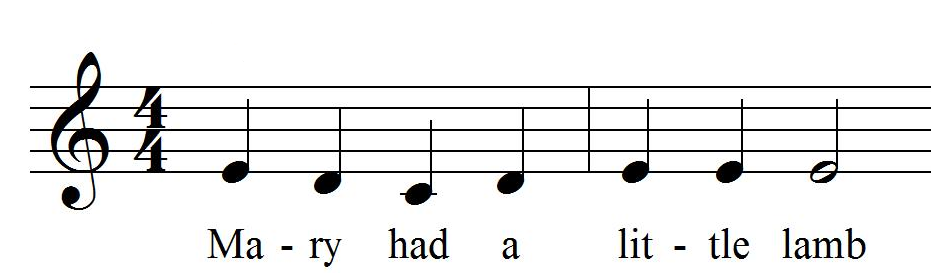}
\caption{
Musical score of ``Mary had a little lamb".
}
\label{fig:mary_representations}
\end{figure} 
All NMF algorithms were run for 200 iterations which allowed them to converge. 
Figure~\ref{fig:compa_mary1} presents the columns of $W$ (dictionary matrix) and the rows of $H$ for baseline KL-NMF and min-vol KL-NMF with $K=3$. 
\begin{figure}
    \centering  
    \begin{subfigure}[b]{0.23\textwidth}
        \includegraphics[width=\textwidth]{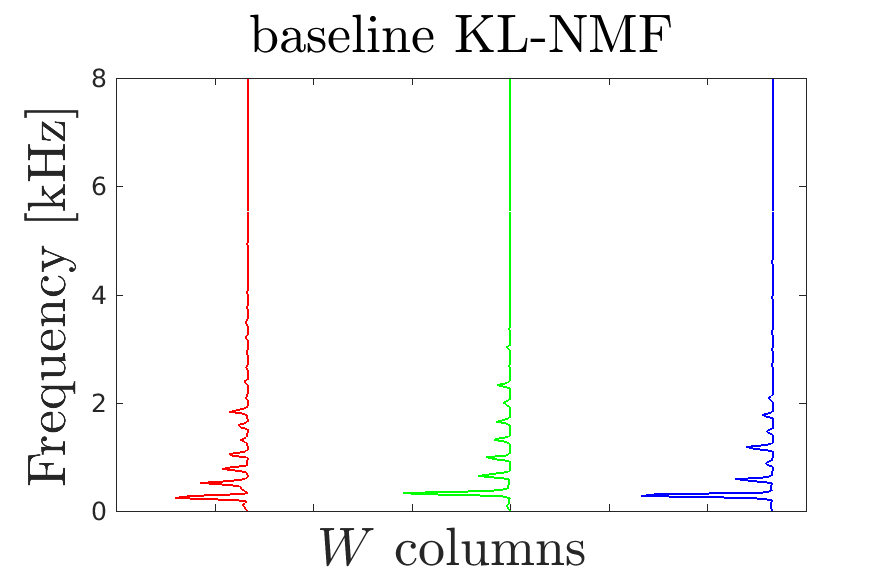}
        \includegraphics[width=\textwidth]{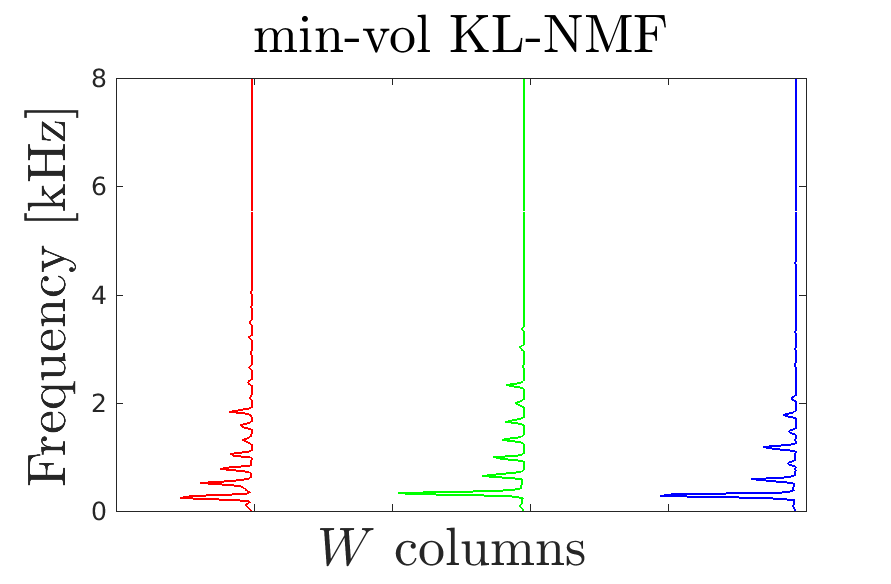}
        \includegraphics[width=\textwidth]{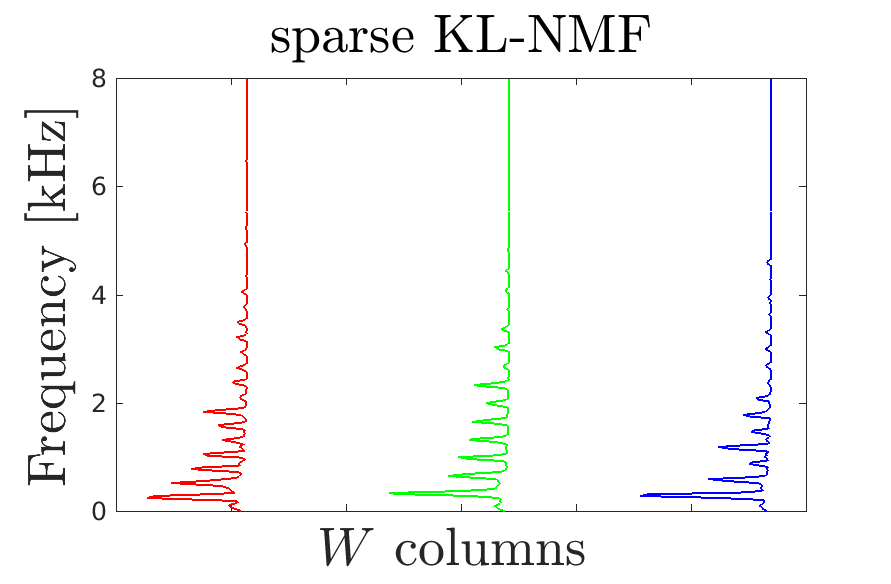}
        \caption{Columns of $W$}
    \end{subfigure}
    ~ 
    \begin{subfigure}[b]{0.23\textwidth}
        \includegraphics[width=\textwidth]{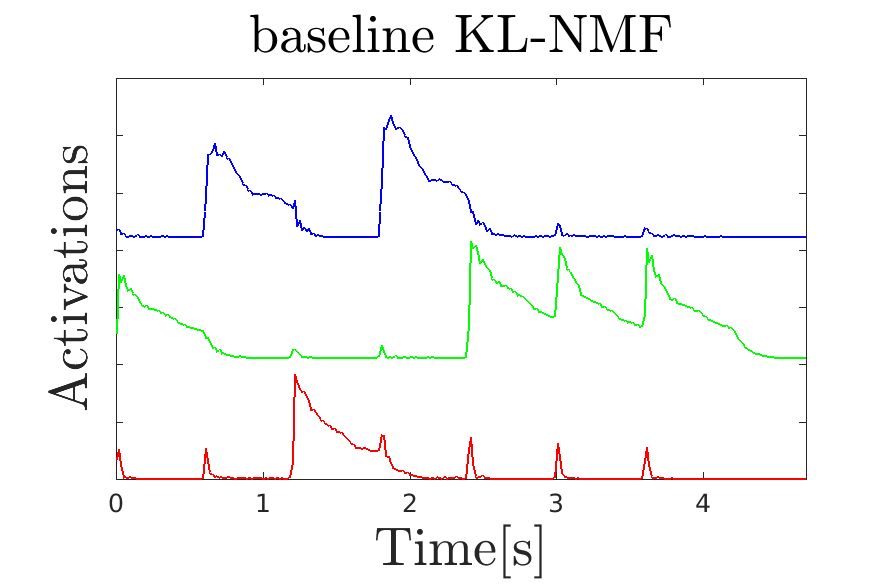}
        \includegraphics[width=\textwidth]{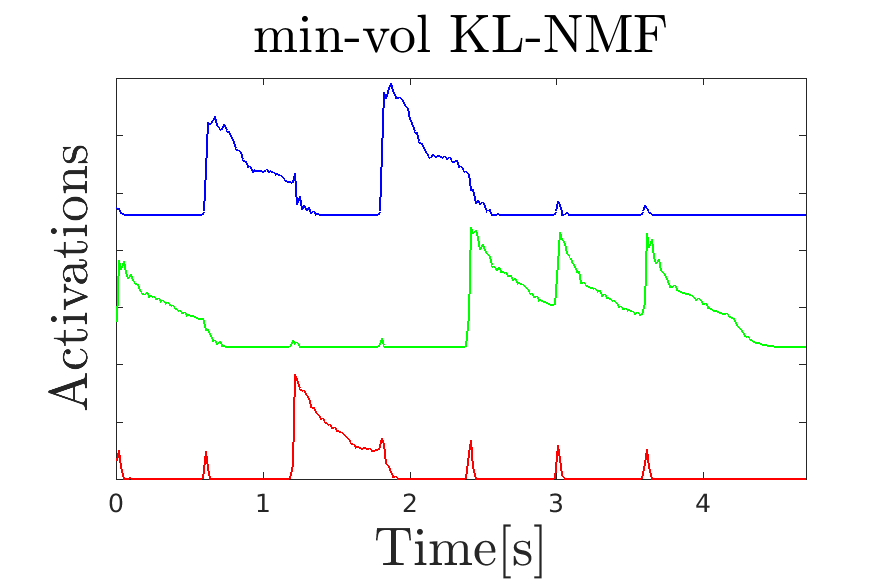}
        \includegraphics[width=\textwidth]{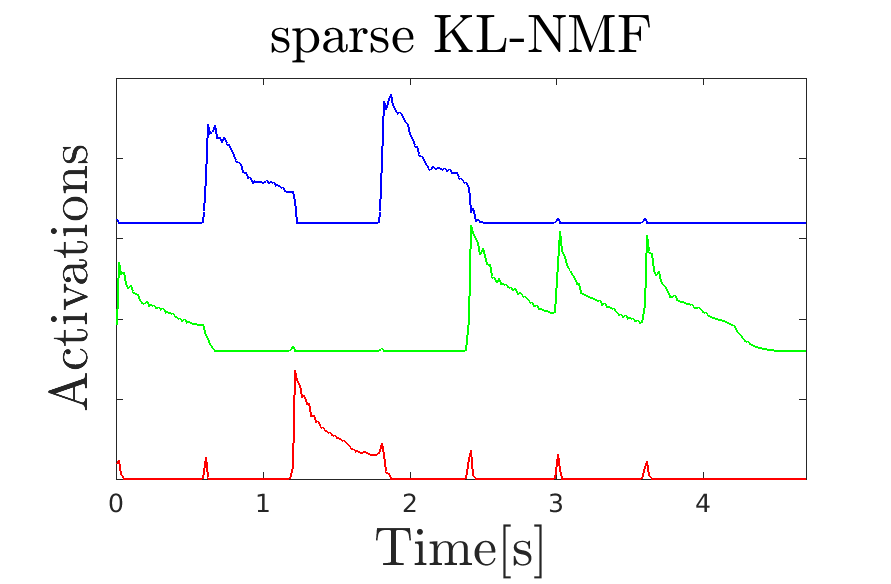}
        \caption{Rows of $H$}
    \end{subfigure}
    \caption{Comparative study of baseline KL-NMF (top), min-vol KL-NMF (middle) and sparse KL-NMF (bottom) applied to ``Mary had a little lamb" amplitude spectrogram with $K$=3.}\label{fig:compa_mary1}
\end{figure} 
Figure~\ref{fig:compa_mary2} presents the time-frequency masking coefficients. These coefficients are computed as follows 
\begin{equation*}
\begin{aligned}
& \text{mask}_{f,n}^{(k)}=\frac{\hat{X}_{f,n}^{(k)}}{\sum_{k}\hat{X}_{f,n}^{(k)}} \qquad \text{with } k=1,...,K \,,
\end{aligned}
\end{equation*}
where $\hat{X}^{(k)} = W(:,k)H(k,:)$ is the estimated source $k$. The masks are nonnegative and sum to one for each pair $\left( f,n \right)$. 
This representation allows to identify visually whether the NMF algorithm was able to separate the sources properly. 
\begin{figure}
    \centering
    \begin{subfigure}[b]{0.3\textwidth}
        \includegraphics[width=\textwidth]{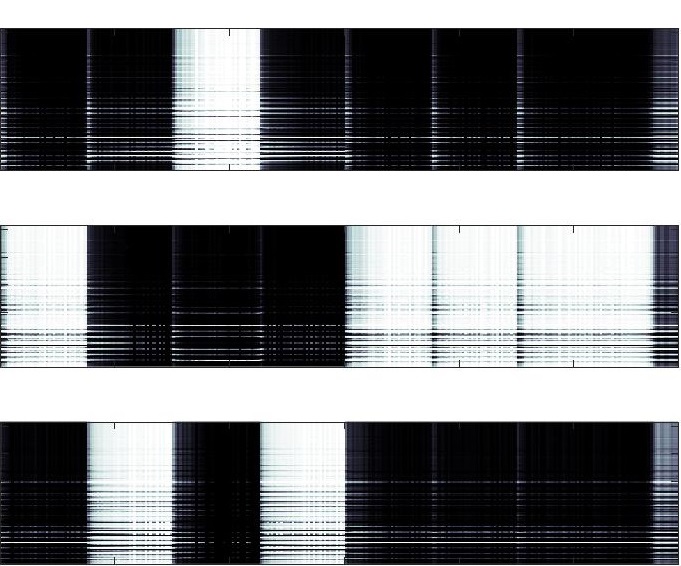}
        \caption{baseline KL-NMF}
        \label{fig:baseline_kl_masks1}
    \end{subfigure}
    ~ 
    \begin{subfigure}[b]{0.3\textwidth}
        \includegraphics[width=\textwidth]{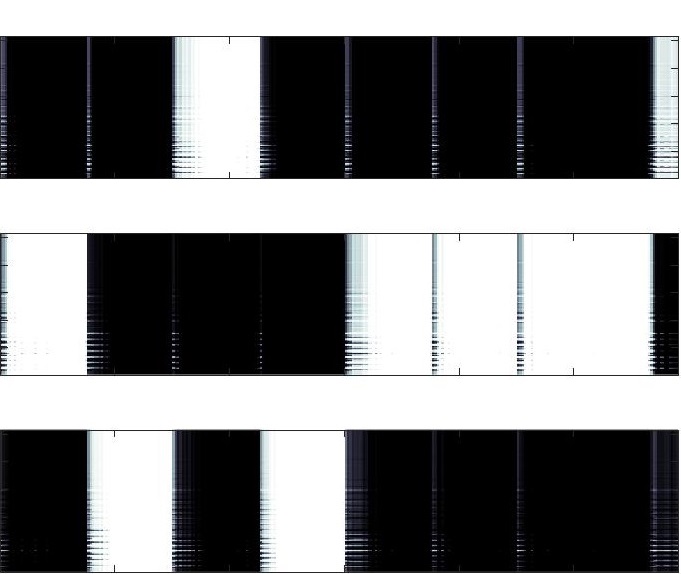}
        \caption{min-vol KL-NMF}
        \label{fig:minvol_kl_masks1}
    \end{subfigure}
    \begin{subfigure}[b]{0.3\textwidth}
        \includegraphics[width=\textwidth]{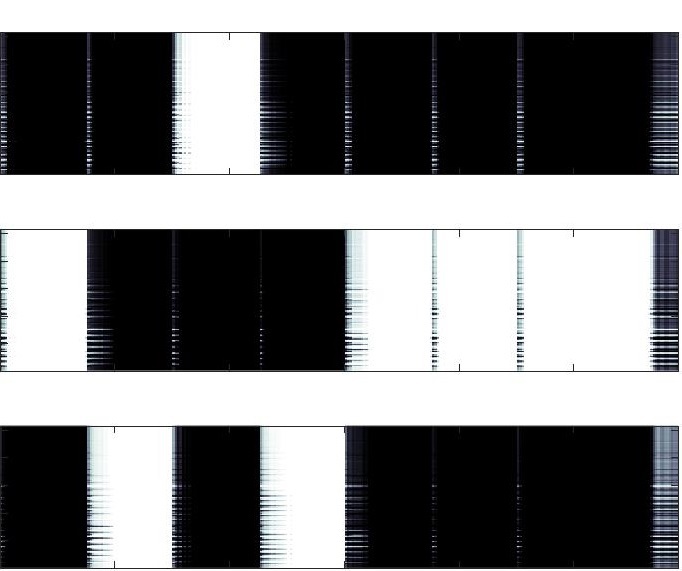}
        \caption{sparse KL-NMF}
        \label{fig:sparse_kl_masks1}
    \end{subfigure}
    \caption{Masking coefficients obtained with baseline KL-NMF (top), min-vol KL-NMF (middle) and sparse KL-NMF (bottom) applied to ``Mary had a little lamb" amplitude spectrogram with $K$=3.}\label{fig:compa_mary2}
\end{figure} 
All the simulations give a nice separation with similar results for $W$ and $H$. The activations are coherent with the sequences of the notes.  However, Figure~\ref{fig:compa_mary2} shows that min-vol KL-NMF and sparse KL-NMF  provide a better separation in terms of time-frequency localization compared to the baseline KL-NMF. 

%
%
%
%
%
%
%
%
%

We now perform the same experiment but using $K$=7. 
Figure~\ref{fig:compa_mary3} presents the results.  
This situation corresponds to the situation where the factorization rank is overestimated. Figure~\ref{fig:compa_mary4} presents the time-frequency masking coefficients.
\begin{figure}
    \begin{subfigure}[b]{0.235\textwidth}
        \includegraphics[width=\textwidth]{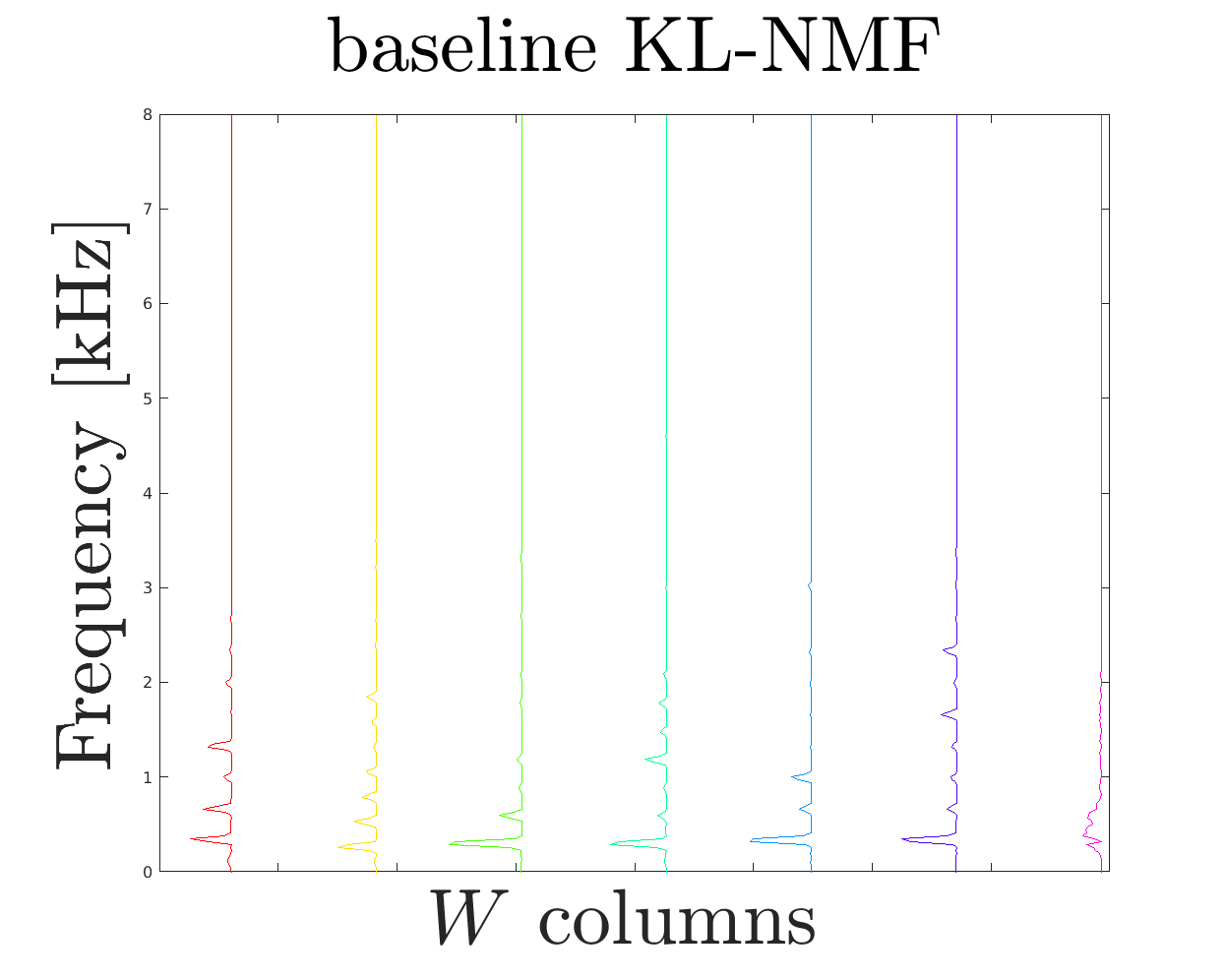}
        \includegraphics[width=\textwidth]{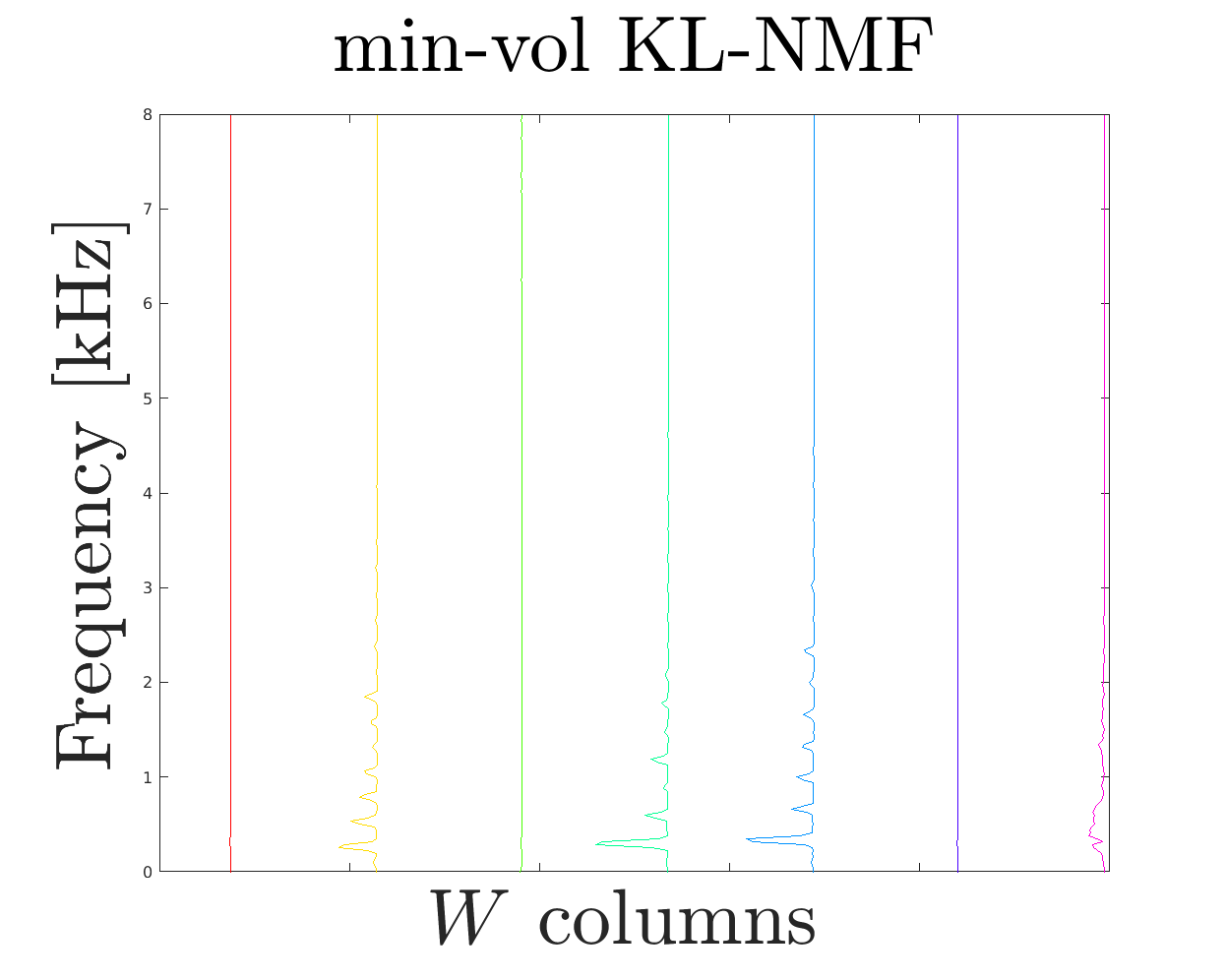}
        \includegraphics[width=\textwidth]{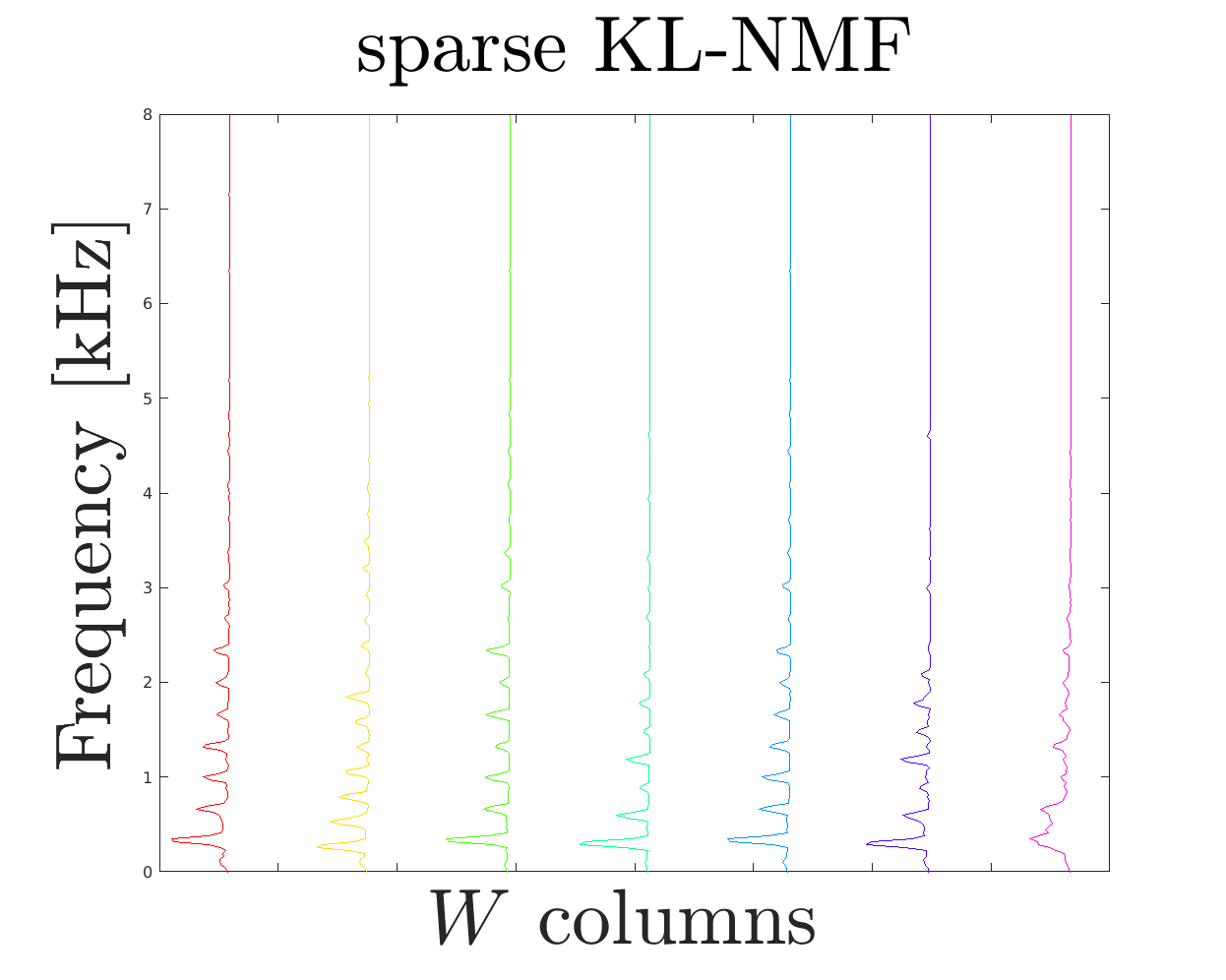}
        \caption{Columns of $W$}
        \label{fig:basis_mary2}
    \end{subfigure}
    ~ 
    \begin{subfigure}[b]{0.235\textwidth}
        \includegraphics[width=\textwidth]{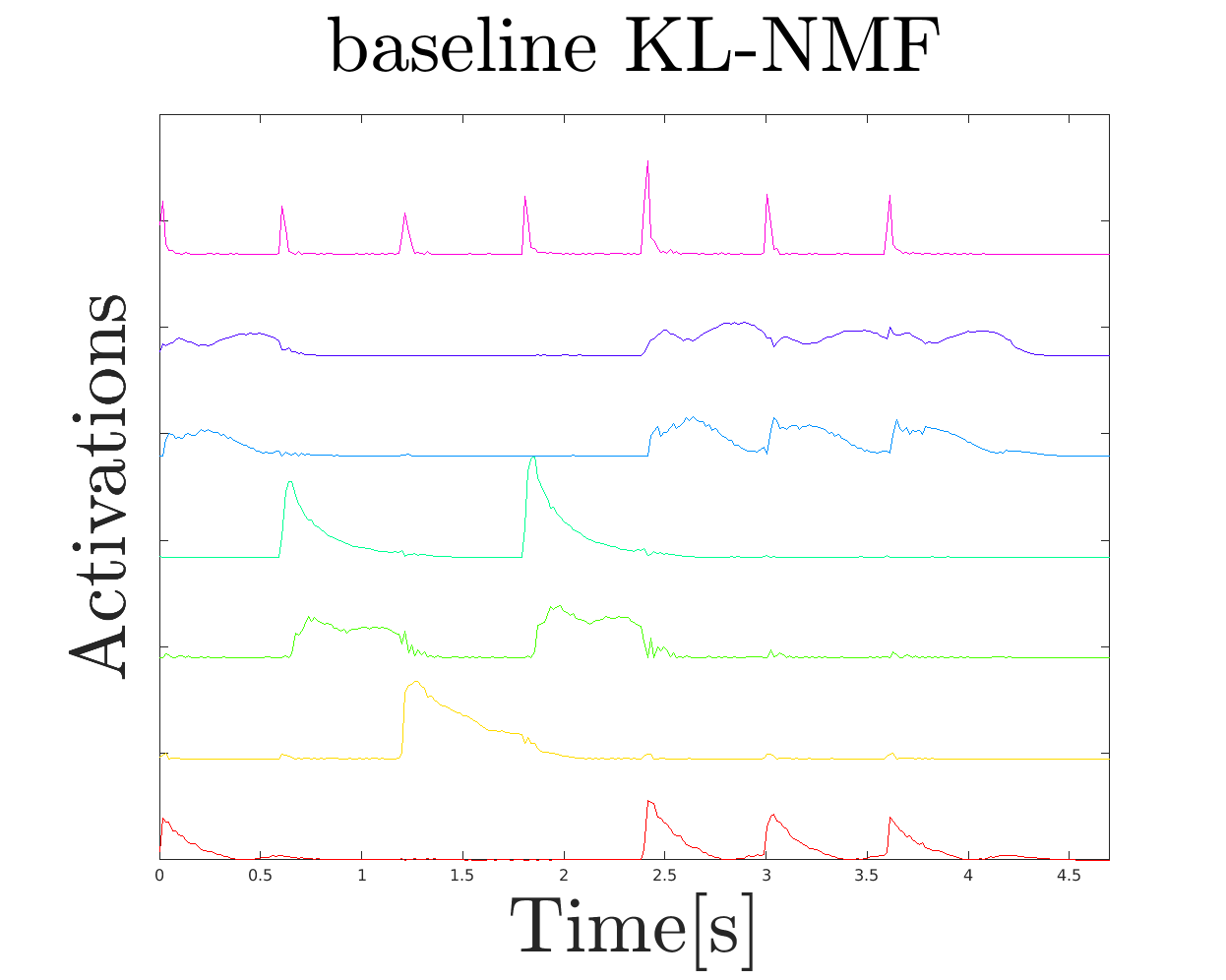}
        \includegraphics[width=\textwidth]{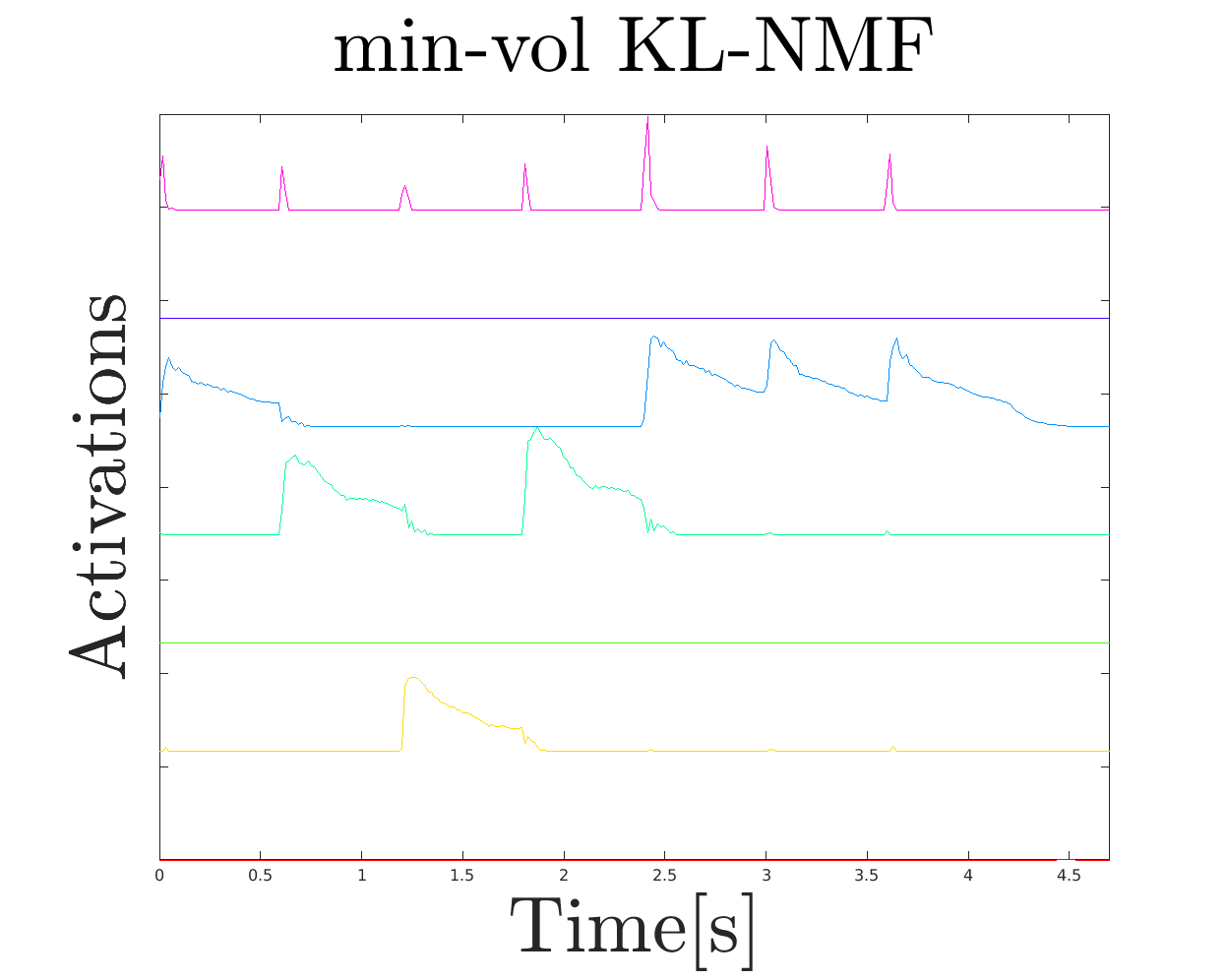}
        \includegraphics[width=\textwidth]{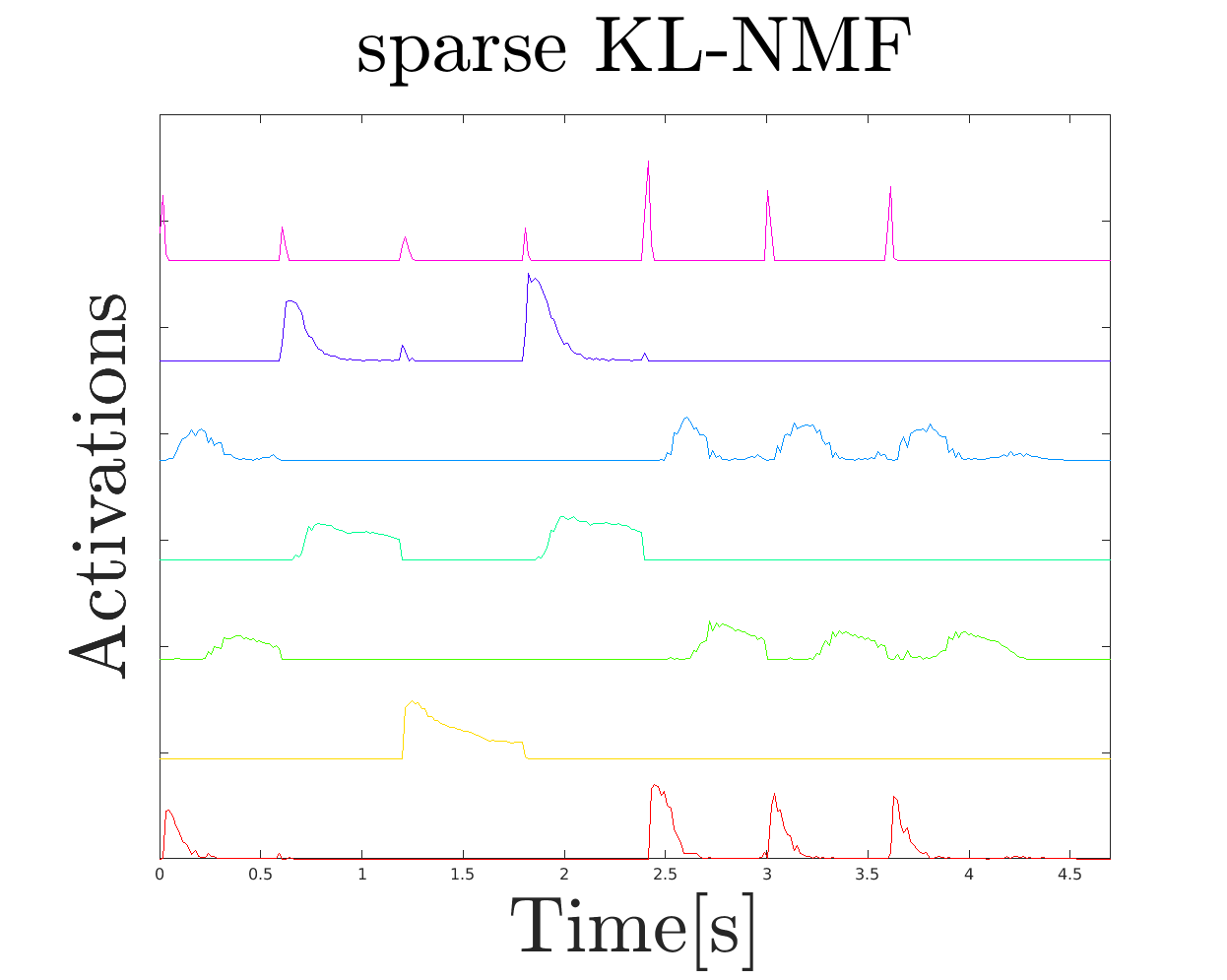}
        \caption{Rows of $H$}
        \label{fig:activations_mary2}
    \end{subfigure}
    \caption{Comparative study of baseline KL-NMF (top), min-vol KL-NMF (middle) and sparse KL-NMF (bottom) applied to ``Mary had a little lamb" amplitude spectrogram with $K$=7. \label{fig:compa_mary3}}
\end{figure}
\begin{figure}
    \centering
    \begin{subfigure}[b]{0.32\textwidth}
        \includegraphics[width=\textwidth]{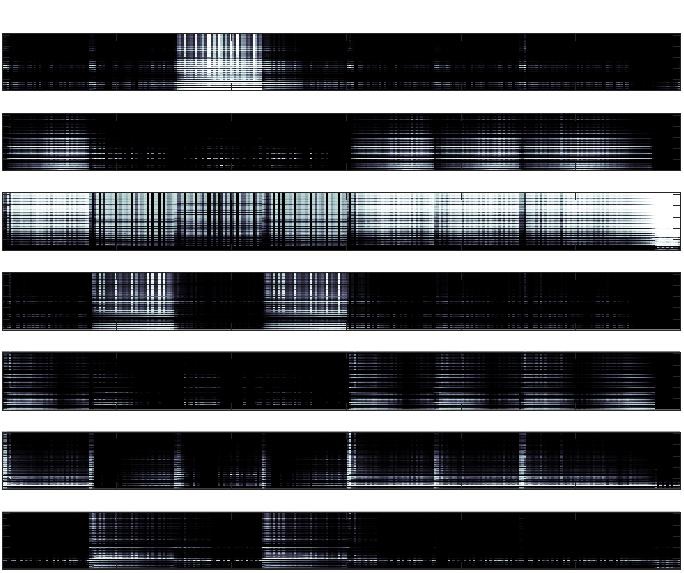}
        \caption{baseline KL-NMF}
        \label{fig:baseline_kl_masks2}
    \end{subfigure}
    ~ 
    \begin{subfigure}[b]{0.32\textwidth}
        \includegraphics[width=\textwidth]{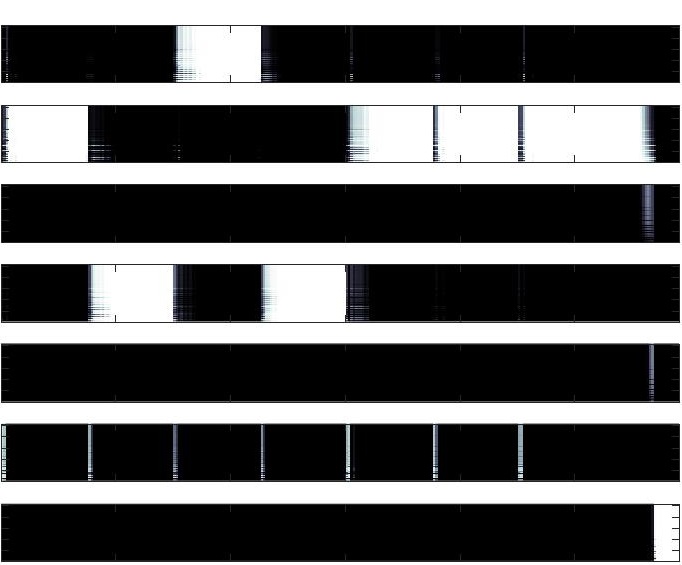}
        \caption{min-vol KL-NMF}
        \label{fig:minvol_kl_masks2}
    \end{subfigure}
    \begin{subfigure}[b]{0.32\textwidth}
        \includegraphics[width=\textwidth]{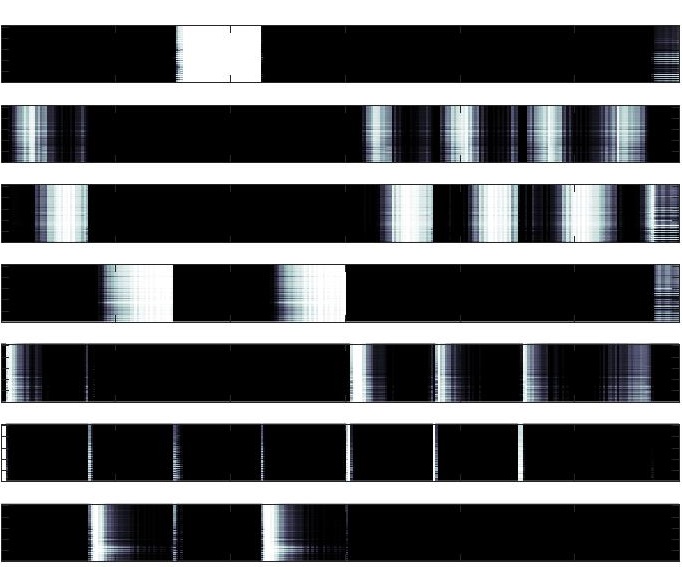}
        \caption{sparse KL-NMF}
        \label{fig:sparse_kl_masks2}
    \end{subfigure}
    \caption{Masking coefficients obtained with baseline KL-NMF (top), min-vol KL-NMF (middle) and sparse KL-NMF (bottom) applied to ``Mary had a little lamb" amplitude spectrogram with $K$=7. \label{fig:compa_mary4}} 
\end{figure} 
We observe that min-vol KL-NMF is able to extract the three notes correctly and set automatically to zero three source estimates (more precisely, three rows of $H$ are set to zero, while the corresponding columns of $W$ have entries equal to one another as $||W(:,k)||_1=1$ for all $k$) while baseline KL-NMF and sparse KL-NMF split the notes in all the sources.  
One can observe that a fourth note is identified in all simulations (see isolated peaks on Figure~\ref{fig:compa_mary4}-(b), second row of $H$ from the top) and corresponds to the noise within the piano just before triggering a specific note (in particular, the hammer noise). This observation is confirmed by the fact that the amplitude is proportional to the natural strength of the fingers playing the notes. 
In this scenario, with $K$ is overestimated, min-vol KL-NMF outperforms baseline KL-NMF and sparse KL-NMF.

\paragraph{Prelude of Bach}\label{paraPrelude} The second audio sample corresponds to the first 30 seconds of ``Prelude and Fugue No.1 in C major" from J.~S.~Bach played by Glenn Gould\footnote{\url{https://www.youtube.com/watch?v=ZlbK5r5mBH4}}. The audio sample is a sequence of 13 notes: $B_{3}$, $C_{4}$, $D_{4}$, $E_{4}$, $F^{\#}_{4}$, $G_{4}$, $A_{4}$, $C_{5}$, $D_{5}$, $E_{5}$, $F_{5}$, $G_{5}$, $A_{5}$. The recorded signal is downsampled to $f_{s}=11025$Hz yielding $T$=330750 samples. STFT of the input signal $x$ yields a temporal resolution of 46ms and a frequency resolution of 10.76Hz, so that the amplitude spectrogram $V$ has $N$=647 frames and $F$=513 frequency bins. 
The musical score is presented on Figure~\ref{fig:prelude_representations}. All NMF algorithms were run for 300 iterations which allowed them to converge.  
Figure~\ref{fig:w_h} presents the results obtained for $W$ and $H$ with a factorization rank $K=16$, hence overestimated. 
We observe that min-vol KL-NMF automatically sets three components to zero (with * symbol on Figure~\ref{fig:w_h}) while 13 source estimates are determined. The analysis of the fundamentals (maximum peak frequency) of the 13 source estimates correspond to the theoretical fundamentals of the 13 notes mentioned earlier. Note that using baseline KL-NMF or sparse KL-NMF led to same conclusions as for the first audio sample; these two algorithms generate as many source estimates as imposed by the rank of factorization while min-vol KL-NMF algorithm preserves the integrity of the 13 sources. Additionally, the activations are coherent with the sequences of the notes.  Figure~\ref{fig:bach_vali_seq} shows (on a limited time interval) that the estimate sequence follows the sequence defined in the score. Note that a threshold and permutations on rows of $H$ was used to improve visibility.

\paragraph{Bass and drums} The third audio signal is a synthetic mix of a bass and drums\footnote{\url{http://isse.sourceforge.net/demos.html}}. 
The audio signal is downsampled to $f_{s}$=$16000$Hz yielding $T$=104821 samples. 
STFT of the input signal $x$ yields a temporal resolution of 32ms and a frequency resolution of 15.62Hz, so that the amplitude spectrogram $V$ has $N$=206 frames and $F$=513 frequency bins. 
 For this synthetic mix, we have access to the true sources under the form of two audio files. Therefore, we can estimate the quality of the separation with standard metrics, namely the signal to distortion ratios (SDR), the source to interference ratios (SIR) and the sources to artifacts ratios (SAR)~\cite{vincent}. They have been computed with the toolbox BSS Eval\footnote{\url{http://bass-db.gforge.inria.fr/bss_eval/}}. 
 The metrics are expressed in dB and the higher they are the better is the separation. Algorithms min-vol KL-NMF, baseline KL-NMF and sparse KL-NMF have been considered for this comparative study. A factorization rank equal to two is used. It is clear that the rank-one approximation is too simplistic for these sources but the goal is to compare the algorithms and show that min-vol KL-NMF is able to find a better solution even in this simplified context. 
All NMF algorithms were run for 400 iterations which allowed them to converge. 
Table~\ref{table:metric_compa} shows the results. 
\begin{center}
\begin{table*}[h!]
\begin{center}
\caption{SDR, SIR and SAR metrics comparison for results obtained with baseline KL-NMF and min-vol KL-NMF on a synthetic mix of bass and drums}
\label{table:metric_compa}
\begin{tabular}{|c|c|c|c|c|c|c|}
  \hline
  Algorithms       & \multicolumn{3}{|c|}{Source 1: bass} & \multicolumn{3}{|c|}{Source 2: drums} \\
                   & SDR(dB) & SIR(dB) & SAR(dB)          & SDR(dB) & SIR(dB) & SAR(dB)           \\
  \hline
  min-vol KL-NMF   & \textbf{-1.14} & \textbf{0.12} & \textbf{7.78}  & \textbf{9.60} & \textbf{19.8}  & 10.09 \\
 
  \hline
  baseline KL-NMF  & -4.26 & -1.39 & 2.64                 & 7.97 & 9.00  & \textbf{15.25} \\

  \hline
  sparse KL-NMF    & -4.69 & -1.73 & 2.33                  & 7.89 & 8.96  & 14.98 \\  \hline
\end{tabular} 
\end{center}
\end{table*}
\end{center}
Except for SAR metric for the second source (drums), min-vol KL-NMF outperforms baseline KL-NMF and sparse KL-NMF.  
\begin{figure}
 \includegraphics[width=0.5\textwidth]{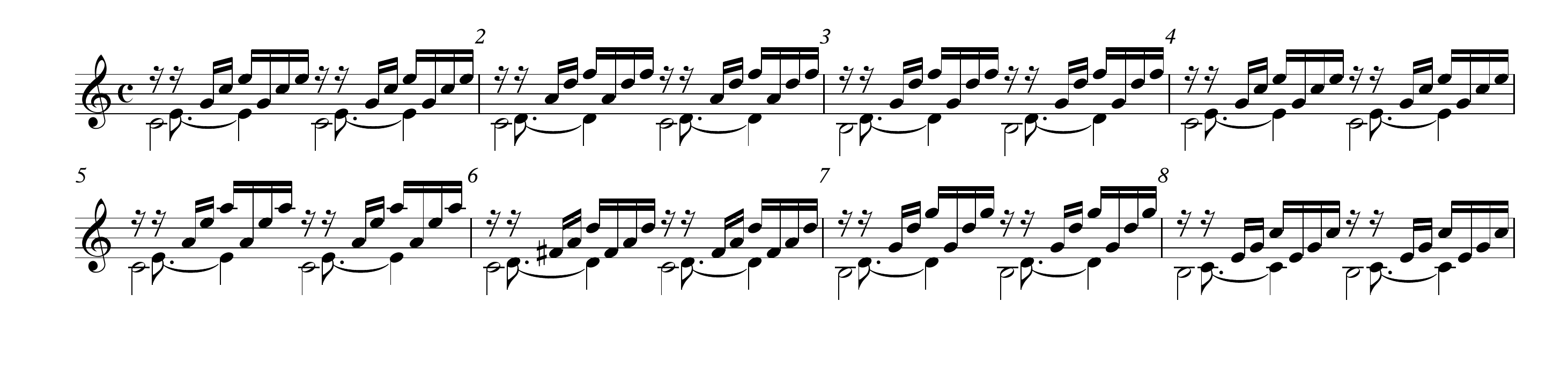} 
\caption{
Musical score of the sample ``Prelude and Fugue No.1 in C major". 
}
\label{fig:prelude_representations}
\end{figure} 
\begin{figure}
\centering
\begin{subfigure}[b]{0.4\textwidth}
        \includegraphics[width=\textwidth]{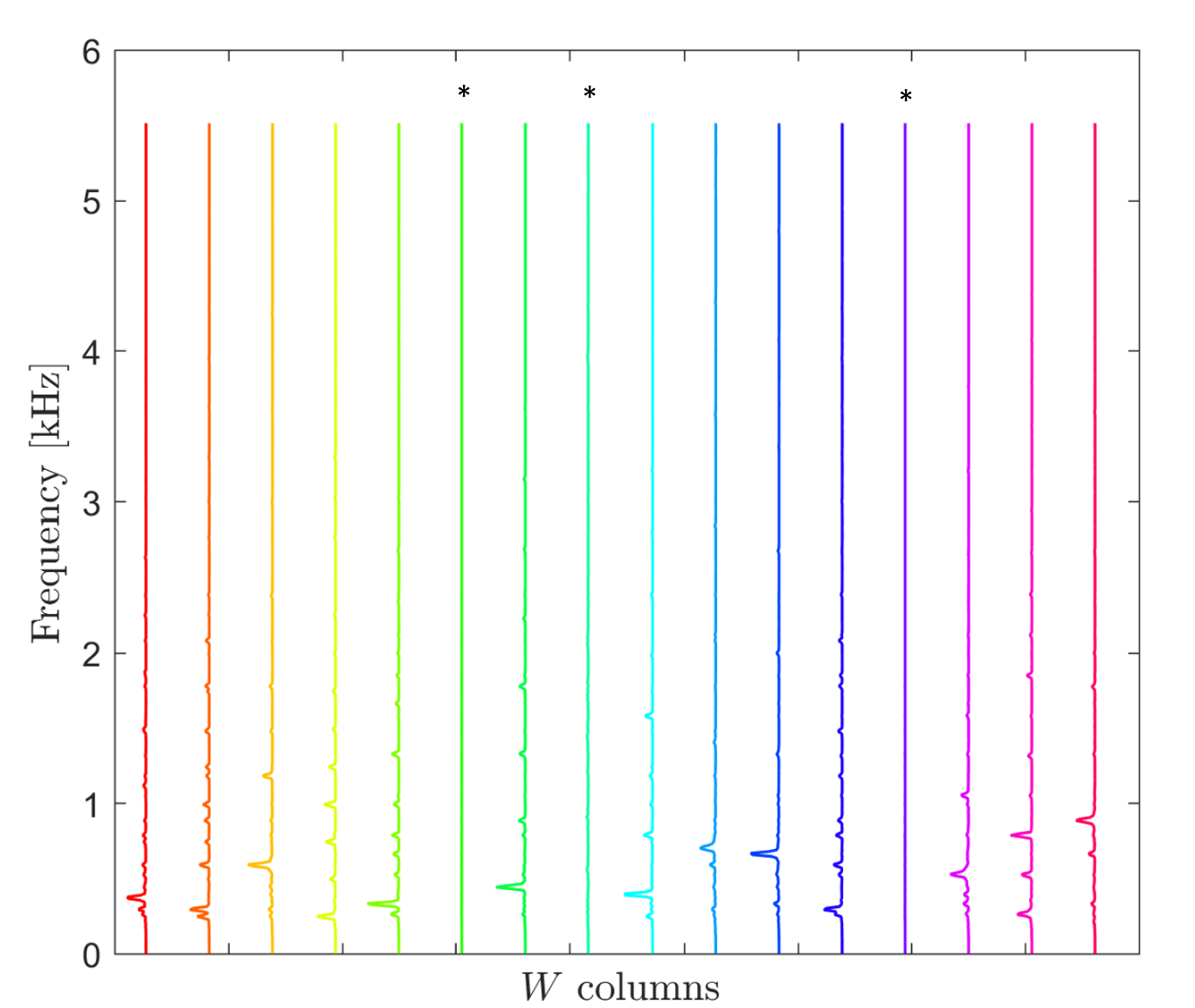}
        \caption{Columns of $W$}
    \end{subfigure}
\begin{subfigure}[b]{0.4\textwidth}
        \includegraphics[width=\textwidth]{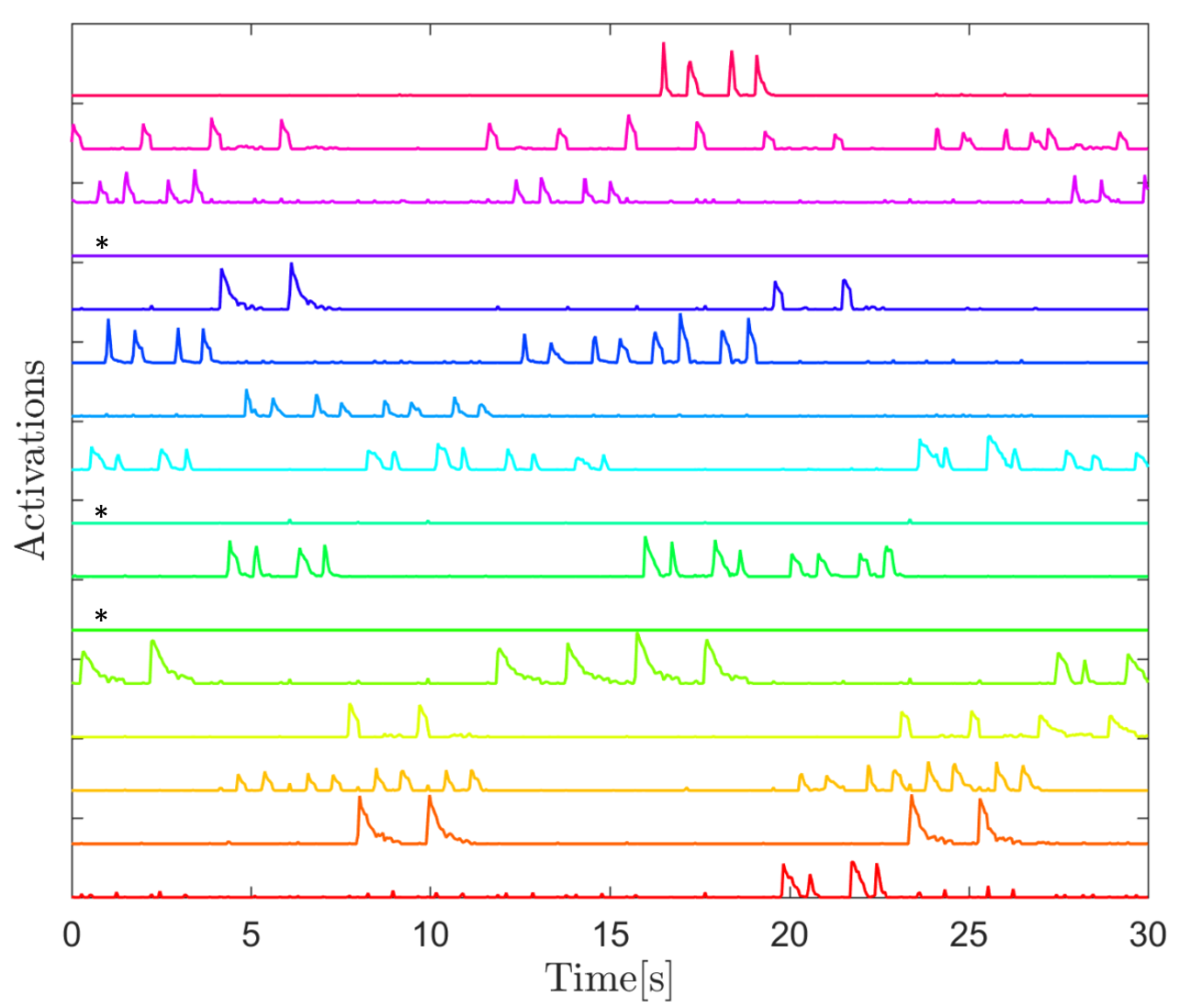}
        \caption{Rows of $H$}
    \end{subfigure}
\caption{
Factors matrices $W$ and $H$ obtained with min-vol KL-NMF with factorization rank $K$=16 on the sample ``Prelude and Fugue No.1 in C major". 
}
\label{fig:w_h}
\end{figure}

\begin{figure}
\centering
\begin{subfigure}[b]{0.5\textwidth}
        \includegraphics[width=\textwidth]{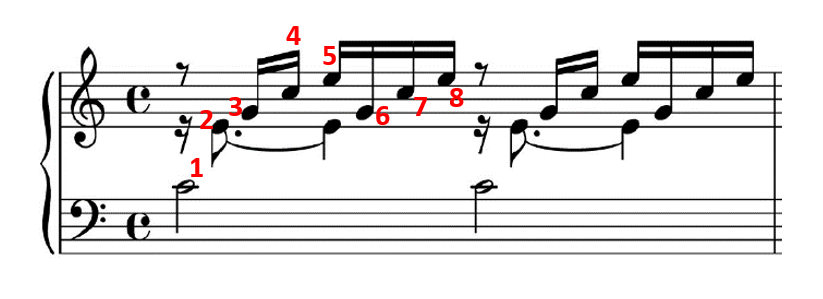}
    \end{subfigure}
\begin{subfigure}[b]{0.45\textwidth}
        \includegraphics[width=\textwidth]{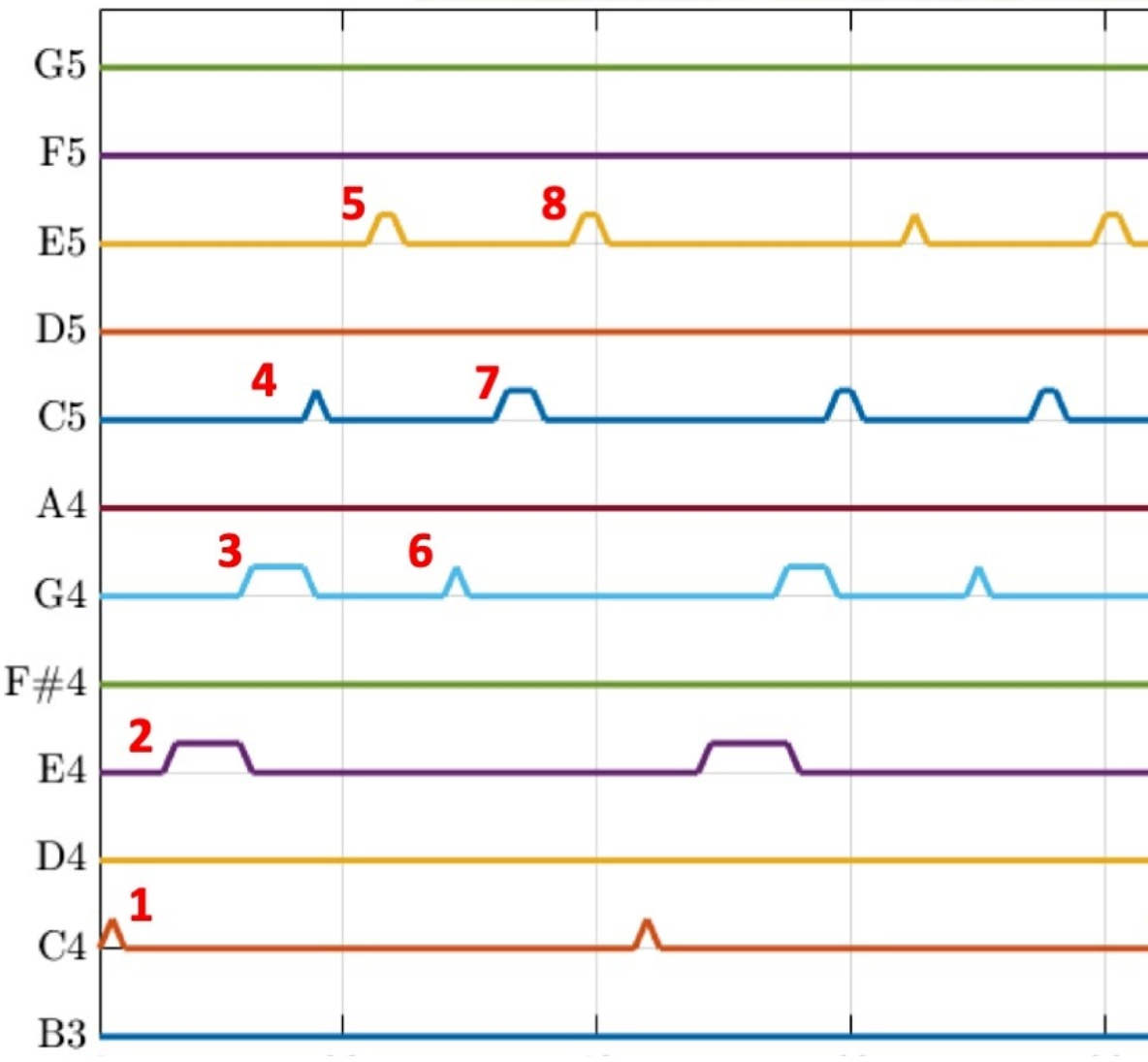}
    \end{subfigure}
\caption{
Validation of the estimate sequence obtained with min-vol KL-NMF with factorization rank $K$=16 on the sample ``Prelude and Fugue No.1 in C major".
}
\label{fig:bach_vali_seq}
\end{figure}

\modif
\paragraph{Runtime performance}
 Let us compare the runtime of baseline KL-NMF, min-vol KL-NMF (Algorithm~\ref{KLminvol}) and sparse KL-NMF~\cite{Leroux}. The algorithms are compares on the three examples presented in paragraphs \ref{paraMary} and \ref{paraPrelude}: 
\begin{itemize}
    \item Setup $\sharp$1: sample “Mary  had  a  little  lamb” with \mbox{$K=3$}, 200 iterations. 
    \item Setup $\sharp$2: sample “Mary  had  a  little  lamb” with \mbox{$K=7$}, 200 iterations. 
    \item Setup $\sharp$3: “Prelude and Fugue No.1 in C major” with $K=16$, 300 iterations.
\end{itemize}
For each test setup, the algorithms are run for the same 20 random initializations of $W$ and $H$. 
Table~\ref{table:runtimeperf} reports the average and standard deviation of the runtime (in seconds) over these 20 runs. 
We observe that the runtime of min-vol KL-NMF (Algorithm~\ref{KLminvol}) is slower but not significantly so, as expected. In particular, on the larger setup~$\sharp$3, it is less than three times slower than the standard MU. 
\begin{center}
\begin{table}[ht!]
\begin{center}
\modif
\caption{\modif Runtime  performance in seconds of baseline KL-NMF, min-vol KL-NMF (Algorithm~\ref{KLminvol}) and sparse KL-NMF~\cite{Leroux}. The table reports the average and standard deviation over 20 random initializations for three experimental setups described in the text. \endmodif } 
\label{table:runtimeperf}
\begin{tabular}{|c|c|c|c|}
  \hline
  Algorithms       & \multicolumn{3}{|c|}{runtime in seconds} \\
                   & setup $\sharp$1 & setup $\sharp$2 & setup $\sharp$3  \\
  \hline
  baseline KL-NMF   & 0.44$\pm$0.03 & 0.43$\pm$0.01 & 3.81$\pm$0.19  \\
 
  \hline
  min-vol KL-NMF  & 3.79$\pm$0.13 & 2.39$\pm$0.30 & 10.19$\pm$1.28                 \\

  \hline
  sparse KL-NMF    & 0.20$\pm$0.02 & 0.20$\pm$0.01 & 2.21$\pm$ 0.26                  \\  \hline
\end{tabular} 
\endmodif 
\end{center}
\end{table}
\end{center}

\endmodif

\section{Conclusion and Perspectives} \label{conclu}

In this paper, we have presented a new NMF model of audio source separation based on the minimization of a cost function that includes a $\beta$-divergence (data fitting term)  and a penalty term that promotes solutions $W$ with minimum volume. We have proved the identifiability of the model in the exact case, under the sufficiently scattered condition for the activation matrix $H$. 
We have provided multiplicative updates to tackle this problem and have illustrated the behaviour of the method on real-world audio signals. We highlighted the capacity of the model to deal with the case where $K$ is overestimated by setting automatically to zero some components and give good results for the source estimates. 

Further work includes tackling the following questions:  
\begin{itemize}
	\item Under which conditions can we prove the identifiability of min-vol $\beta$-NMF in the presence of noise, and the rank-deficient case?
	
	
	\item Can we prove that min-vol $\beta$-NMF performs model order selection automatically? Under which conditions? 
	We have observed this behaviour on many examples, but the proof remains elusive. 
	
	\item Can we design more efficient algorithms? 
\end{itemize}
Further work also includes the use of our new model min-vol $\beta$-NMF for other applications and the design of more efficient algorithms (for example, that avoid using a line-search procedure) with stronger convergence guarantees (beyond the monotonicity of the objective function).

 \paragraph*{Acknowledgments} We thank Kejun Huang and Xiao Fu for helpful discussion on Theorem~\ref{mainth}, and giving us the insight to adapt their proof from~\cite{nmfidentifiable} to our model~\eqref{eq:4}.  
 \modif We also thank the reviewers for their insightful comments that helped us improve the paper. 
\endmodif

\appendix

\subsection{Sufficiently scattered condition and identifiability} \label{app:iden}  

 
 Before giving the definition of the sufficiently scattered condition from~\cite{huang},  
 let us first recall an important property of the duals of nested cones. 
\begin{lemma}\label{lemma2}
Let $\mathcal{C}_{1}$ and $\mathcal{C}_{2}$ be convex cones such that $\mathcal{C}_{1} \subseteq \mathcal{C}_{2}$. 
Then $\mathcal{C}^{*}_{2} \subseteq \mathcal{C}^{*}_{1}$ where $\mathcal{C}^{*}_{2}$ and $\mathcal{C}^{*}_{1}$ are respectively the dual cones of $\mathcal{C}_{1}$ and $\mathcal{C}_{2}$. The dual of a cone $\mathcal{C}$  is defined as 
$\mathcal{C}^{*}=\left\lbrace y| x^T y \geq 0 \text{ for all } x \in \mathcal{C}  \right\rbrace$.
\end{lemma}

\begin{definition}\label{def4}
(Sufficiently Scattered) A matrix $H \in \mathbb{R}_{+}^{K \times N}$ is sufficiently scattered if 
\begin{enumerate}
	\item $\mathcal{C} \subseteq \cone \left( H \right) $, and  
	\item $\cone \left( H \right)^{*} \cap \text{bd}\mathcal{C}^{*}=\left\lbrace\lambda e_{k}|\lambda \geq 0, k=1,...,K\right\rbrace$, 
\end{enumerate}
where $\mathcal{C} = \left\lbrace x|x^{T}e \geq \sqrt{K-1} \left\| x \right\|_{2}   \right\rbrace$ is a second order cone, $\mathcal{C}^{*} = \left\lbrace x|x^{T}e \geq \left\| x \right\|_{2}   \right\rbrace$, 
$\cone\left( H \right)=\left\lbrace x|x=H\theta,  \theta \geq 0   \right\rbrace$ is the conic hull of the columns of $H$, and 
$\text{bd}$ denotes the boundary of a set.
\end{definition} 

We can now prove Theorem~\ref{mainth}. 

\begin{proof}[Proof of Theorem~\ref{mainth}] 
Recall that $W^{\#}$ and $H^{\#}$ are the true latent factors that generated $V$, with $\text{rank}(V) = K$ and $H^{\#}$ is sufficiently scattered. 
Let us consider $\hat{W}$ and $\hat{H}$ 
a feasible solution of~\eqref{eq:8b}. 
Since $\text{rank}(V)=K$ and $V=\hat{W}\hat{H}$, 
we must have $\rank(\hat{W})=\text{rank}(\hat{H})=K$. 
Hence there exists an invertible matrix $A \in \mathbb{R}^{K \times K}$ such that $\hat{W}=W^{\#}A^{-1}$ and $\hat{H}=A H^{\#}$.  
Since $\hat{W}$ is a feasible solution of problem~\eqref{eq:8b}, we have  
\begin{equation*} 
\begin{aligned}
&  e^{T}\hat{W}=e^{T}W^{\#}A^{-1}=e^{T}A^{-1}=e^{T}, 
\end{aligned}
\end{equation*} 
where we assumed $e^{T}W^{\#}=e^{T}$ without loss of generality since $W^{\#} \geq 0$ and $\text{rank}(W^{\#})=K$.  
Note that $e^{T}A^{-1}=e^{T}$ is equivalent to $e^{T}A=e^{T}$. This means that matrix $A$ is column stochastic. 
Therefore we have that $e^{T}A e=K$.
Since $\hat{H}$ is a feasible solution, we also have 
$\hat{H}=A H^{\#} \geq 0$. 
Let us denote by $a_{j}$ the $j$th row of A, and by $a^{T}_{k}$ the $k$th column of $A^{T}$.  
By the definition of the a dual cone, $A H^{\#} \geq 0$ means that the rows $a_{j} \in \cone(H^{\#})^{*}$ for $j=1,...,K$. 
Since $H^{\#}$ is sufficiently scattered, $\cone\left( H \right)^{*} \subseteq \mathcal{C}^{*}$ (by Lemma~\ref{lemma2}) hence $a_{j} \in \mathcal{C}^{*}$. 
Therefore we have $\left\| a_{j} \right\|_{2} \leq a_{j} e$ by definition of $\mathcal{C}$. 
This leads to the following: 
  $|\text{det}(A)|=|\text{det}(A^{T})| 
  \leq \prod_{k = 1}^{K} \left\| a^{T}_{k} \right\|_{2} 
                                      = \prod_{j = 1}^{K} \left\| a_j \right\|_{2}  
                                      \leq \prod_{j = 1}^{K} a_j e            
                                      \leq \left( \frac{\sum_{j}^{K} a_j e}{K} \right)^{K} 
                                      = \left( \frac{e^{T}A e}{K}  \right)^{K} 
                                      =1$.  
The first inequality is the Hadamard inequality, the second inequality is due to $a_{j} \in \mathcal{C}^{*}$, the third inequality is the arithmetic-geometric mean inequality. Now we can conclude exactly as is done in \cite[Theorem 1]{nmfidentifiable} by showing that matrix $A$ can only be a permutation matrix for an optimal solution ($\hat{W}$,$\hat{H}$) 
of~\eqref{eq:8b}, and therefore identifiability for model~\eqref{eq:8b} holds.  
\end{proof}

\subsection{Proof of Lemma~\ref{lem1quad}} \label{app:lem3}

Separability of $\bar{l}(w|\tilde{w})$ holds since $\modif \Phi\left( \tilde{w} \right) \endmodif $ is diagonal.
The condition $\bar{l}(\tilde{w}|\tilde{w}) = l(\tilde{w})$ from Definition~\ref{def2} can be checked easily. 
It remains to prove that $\bar{l}(w|\tilde{w})\geq l(w)$ for all $w$. 
Let us first rewrite the quadratic function $l(w)$ using its Taylor expansion at $w = \tilde{w}$: 
$l(w) 
= l(\tilde{w})+\left( w-\tilde{w}\right)^{T}  \nabla l \left( \tilde{w}\right) + 
\frac{1}{2} \left( w-\tilde{w}\right)^{T} \nabla^{2}l\left( \tilde{w} \right)  \left( w-\tilde{w}\right) = l(\tilde{w})+\left( w-\tilde{w}\right)^{T}  \modif 2 \endmodif Y \tilde{w} + \frac{1}{2} 
\left( w-\tilde{w}\right)^{T} 2Y  \left( w-\tilde{w}\right)$.  
Proving that $\bar{l}(w|\tilde{w})\geq l(w)$ is equivalent to proving that 
$\frac{1}{2} \left( w-\tilde{w}\right)^{T} \left[ \modif \Phi\left( \tilde{w} \right) \endmodif  - 2Y \right] \left( w-\tilde{w}\right) \geq 0$, 
which boils down to proving that the matrix $\left[ \modif \Phi\left( \tilde{w} \right) \endmodif  - 2Y \right]$ is positive semi-definite. We have 
$\modif \Phi_{ij}(\tilde{w}) \endmodif= 2 \delta_{ij} \frac{(Y^{+}\tilde{w})_{i}+(Y^{-}\tilde{w})_{i}}{\tilde{w}_{i}}$, 
where $\delta_{ij}$ is the Kronecker symbol. 
Let us consider the following matrix: 
$M_{ij}(\tilde{w})= \tilde{w}_{i} \left[ \modif \Phi\left( \tilde{w} \right) \endmodif  - 2Y \right]_{ij} \tilde{w}_{j}$, 
which is a rescaling of $\left[ \modif \Phi\left( \tilde{w} \right) \endmodif  - 2Y \right]$. It remains to show that $M$ is positive semi-definite\footnote{The remainder of the proof was suggested to us by one of the reviewers, it is more elegant and simpler than our original proof.}. 
Since $M$ is symmetric and its diagonal entries
are non-negative, it is sufficient to show that $M$ is diagonally dominant~\cite[Proposition~7.2.3]{horn1985matrix}, that is, 
\begin{equation*} 
\begin{aligned}
&  \left|M_{ii} \right| \geq \sum_{j \neq i}\left|M_{ij} \right| \; \text{ for all } i. 
\end{aligned}
\end{equation*} 
We have for all $i$ that 
\begin{align*}
M_{ii} & =2w_{i}\sum_{j}\left( Y_{ij}^{+}+Y_{ij}^{-}\right)w_{j}-2w_{i}Y_{ii}w_{i}, \text{ and } \\
 M_{ij} & =-2w_{i}Y_{ij}w_{j} \quad \text{for } j \neq i. 
\end{align*} 
Since $Y_{ij}^{+}+Y_{ij}^{-}=\left|Y_{ij} \right|$, we have 
\begin{equation*} 
\begin{aligned}
  M_{ii}-\sum_{j \neq i}\left|M_{ij} \right|&=2w_{i}\sum_{j}\left|Y_{ij} \right|w_{j}-2w_{i}Y_{ii}w_{i}\\
  & \qquad -2w_{i}\sum_{j \neq i}\left|Y_{ij}\right|w_{j}\\
  & = 2w_{i}\left|Y_{ii} \right|w_{i}-2w_{i}Y_{ii}w_{i} \geq 0, 
\end{aligned}
\end{equation*} 
implying that $M$ is diagonally dominant.

\subsection{Algorithm for min-vol IS-NMF} \label{app:is}

For $\beta=0$ (IS divergence), the derivative of the auxiliary function $\bar{F}(w|\tilde{w})$ with respect to a specific coefficient $w_{k}$ is given by: 
\begin{equation*}\label{eq:48}
\begin{aligned}
\nabla_{w_{k}} \bar{F}(w|\tilde{w})&=\sum_{n}\frac{h_{kn}}{\tilde{v}_{n}}-\sum_{n}h_{kn}\frac{\tilde{w}_{k}^{2} v_{n}}{w_{k}^{2}\tilde{v}_{n}^{2}}+2\lambda \left[ Y \tilde{w} \right]_{k} \\ 
& + 2 \lambda \left[ \diag \left( \frac{Y^{+}\tilde{w}+Y^{-}\tilde{w}}{\tilde{w}} \right) \right]_{k} w_{k}\\
&- 2 \lambda \left[ \diag \left( \frac{Y^{+}\tilde{w}+Y^{-}\tilde{w}}{\tilde{w}} \right) \right]_{k} \tilde{w}_{k}. 
\end{aligned}
\end{equation*}
Let  
\begin{equation}\label{eq:49}
\begin{aligned}
&\tilde{a}= 2 \lambda \left[ \diag \left( \frac{Y^{+}\tilde{w}+Y^{-}\tilde{w}}{\tilde{w}} \right) \right]_{k}, \\  
&\tilde{b} = \sum_{n}\frac{h_{kn}}{\tilde{v}_{n}} - 4\lambda \left[ Y^{-} \tilde{w} \right]_{k},  \\
&\tilde{d}= -\sum_{n}h_{kn}\frac{\tilde{w}_{k}^{2} v_{n}}{\tilde{v}_{n}^{2}}. \\
\end{aligned}
\end{equation} 
Setting the derivative to zero requires to compute the roots of the following degree-three polynomial $\tilde{a} w_{k}^{3} + {\tilde{b}} w_{k}^{2} + {\tilde{d}}$.  
We used the procedure developed in \cite{root3} which is based on the explicit calculation of the intermediary root of a canonical form of cubic. This procedure is able to provide highly accurate numerical results even for badly conditioned polynomials.  
The algorithm for min-vol IS-NMF follows the same steps as for min-vol KL-NMF: only the two steps corresponding to the updates of $W$ and $H$ have to be modified. 
For the update of $H$ (step 4), use the standard MU. 
For the update of $W$ (step 9), use \\ 
\noindent \textbf{for} {$f \gets 1$ to $F$}   \\ 
	 {${}$}  \hspace{0.2cm}   \textbf{for} {$k \gets 1$ to $K$}\\ 
	 {${}$}  \hspace{0.4cm} $\text{Compute } \tilde{a} \text{, } \tilde{b} \text{ and } \tilde{d} \text{ according to equations \eqref{eq:49}}$\\  
	 {${}$}  \hspace{0.4cm} 
	Compute the roots of $\tilde{a} w_{k}^{3} + {\tilde{b}} w_{k}^{2} + {\tilde{d}}$ \\  
	 {${}$}  \hspace{0.4cm} Pick $y$ among these roots and zero that minimizes  \\  {${}$}  \hspace{0.4cm} the objective  \\  
	 {${}$}  \hspace{0.4cm} $W^{+}_{f,k} \leftarrow \text{max} \left(10^{-16},y\right)$\\  
	 {${}$}  \hspace{0.2cm}  \textbf{end for}     \\             
    \textbf{end for}


\ifCLASSOPTIONcaptionsoff
  \newpage
\fi

\bibliographystyle{IEEEtran}
\bibliography{Article2018}

\begin{thebibliography}{10}
\providecommand{\url}[1]{#1}
\csname url@samestyle\endcsname
\providecommand{\newblock}{\relax}
\providecommand{\bibinfo}[2]{#2}
\providecommand{\BIBentrySTDinterwordspacing}{\spaceskip=0pt\relax}
\providecommand{\BIBentryALTinterwordstretchfactor}{4}
\providecommand{\BIBentryALTinterwordspacing}{\spaceskip=\fontdimen2\font plus
\BIBentryALTinterwordstretchfactor\fontdimen3\font minus
  \fontdimen4\font\relax}
\providecommand{\BIBforeignlanguage}[2]{{%
\expandafter\ifx\csname l@#1\endcsname\relax
\typeout{** WARNING: IEEEtran.bst: No hyphenation pattern has been}%
\typeout{** loaded for the language `#1'. Using the pattern for}%
\typeout{** the default language instead.}%
\else
\language=\csname l@#1\endcsname
\fi
#2}}
\providecommand{\BIBdecl}{\relax}
\BIBdecl

\bibitem{Lefevre_phd}
A.~Lef\`{e}vre, ``M\'{e}thode d'apprentissage de dictionnaire pour la
  s\'{e}paration de sources audio avec un seul capteur,'' Ph.D. dissertation,
  Ecole Normale Sup\'{e}rieure de Cachan, 2012.

\bibitem{Magron_phd}
P.~Magron, ``Reconstruction de phase par mod\`{e}les de signaux : application
  à la s\'{e}paration de sources audio,'' Ph.D. dissertation, TELECOM
  ParisTech, 2016.

\bibitem{algoNMFlee}
D.~Lee and H.~Seung, ``Algorithms for non-negative matrix factorization,'' in
  \emph{NIPS'00 Proceedings of the 13th International Conference on Neural
  Information Processing Systems}, NIPS.\hskip 1em plus 0.5em minus 0.4em\relax
  MIT Press Cambridge, 2000, pp. 535--541.

\bibitem{fevotte2009nonnegative}
C.~F{\'e}votte, N.~Bertin, and J.-L. Durrieu, ``Nonnegative matrix
  factorization with the {Itakura-Saito} divergence: With application to music
  analysis,'' \emph{Neural computation}, vol.~21, no.~3, pp. 793--830, 2009.

\bibitem{Fevotte_betadiv}
C.~F{\'e}votte and J.~Idier, ``Algorithms for nonnegative matrix factorization
  with the $\beta$-divergence,'' \emph{Neural computation}, vol.~23, no.~9, pp.
  2421--2456, 2011.

\bibitem{zhou2011minimum}
G.~Zhou, S.~Xie, Z.~Yang, J.-M. Yang, and Z.~He, ``Minimum-volume-constrained
  nonnegative matrix factorization: Enhanced ability of learning parts,''
  \emph{IEEE Transactions on Neural Networks}, vol.~22, no.~10, pp. 1626--1637,
  2011.

\bibitem{nmfidentifiable}
X.~Fu, K.~Huang, and N.~D. Sidiropoulos, ``On identifiability of nonnegative
  matrix factorization,'' \emph{IEEE Signal Processing Letters}, vol.~25,
  no.~3, pp. 328--332, 2018.

\bibitem{huang}
K.~Huang, N.~Sidiropoulos, and A.~Swami, ``Non-negative matrix factorization
  revisited: Uniqueness and algorithm for symmetric decomposition,'' \emph{IEEE
  Transactions on Signal Processing}, vol.~62, no.~1, pp. 211--224, 2014.

\bibitem{fu2019nonnegative}
X.~Fu, K.~Huang, N.~Sidiropoulos, and W.-K. Ma, ``Nonnegative matrix
  factorization for signal and data analytics: Identifiability, algorithms, and
  applications,'' \emph{IEEE Signal Processing Magazine}, vol.~36, pp. 59--80,
  2019.

\bibitem{robustvol}
X.~Fu, K.~Huang, B.~Yang, W.-K. Ma, and N.~Sidiropoulos, ``Robust volume
  minimization-based matrix factorization for remote sensing and document
  clustering,'' \emph{IEEE Trans. on Signal Processing}, vol.~64, no.~23, p.
  6254–6268.

\bibitem{ang2019algorithms}
A.~Ang and N.~Gillis, ``Algorithms and comparisons of non-negative matrix
  factorization with volume regularization for hyperspectral unmixing,''
  \emph{Journal of Selected Topics in Applied Earth Observations and Remote
  Sensing}, 2019, to appear.

\bibitem{volminleplat}
V.~Leplat, A.~Ang, and N.~Gillis, ``Minimum-volume rank-deficient nonnegative
  matrix factorizations,'' in \emph{IEEE International Conference on Acoustics,
  Speech and Signal Processing (ICASSP)}.\hskip 1em plus 0.5em minus
  0.4em\relax IEEE, 2019, pp. 3402--3406.

\bibitem{lin2015identifiability}
C.-H. Lin, W.-K. Ma, W.-C. Li, C.-Y. Chi, and A.~Ambikapathi, ``Identifiability
  of the simplex volume minimization criterion for blind hyperspectral
  unmixing: The no-pure-pixel case,'' \emph{IEEE Transactions on Geoscience and
  Remote Sensing}, vol.~53, no.~10, pp. 5530--5546, 2015.

\bibitem{arora2016computing}
S.~Arora, R.~Ge, R.~Kannan, and A.~Moitra, ``Computing a nonnegative matrix
  factorization---provably,'' \emph{SIAM Journal on Computing}, vol.~45, no.~4,
  pp. 1582--1611, 2016.

\bibitem{NMF_complexity}
S.~Vavasis, ``On the complexity of nonnegative matrix factorization,''
  \emph{SIAM Journal on Optimization}, vol.~20, no.~3, pp. 1364--1377, 2010.

\bibitem{fu2018anchor}
X.~Fu, K.~Huang, N.~D. Sidiropoulos, Q.~Shi, and M.~Hong, ``Anchor-free
  correlated topic modeling,'' \emph{IEEE Transactions on Pattern Analysis and
  Machine Intelligence}, vol.~41, no.~5, pp. 1056--1071, 2019.

\bibitem{sun2017majorization}
Y.~Sun, P.~Babu, and D.~Palomar, ``Majorization-minimization algorithms in
  signal processing, communications, and machine learning,'' \emph{IEEE
  Transactions on Signal Processing}, vol.~65, no.~3, pp. 794--816, 2017.

\bibitem{Leroux}
J.~L. Roux, F.~J. Weninger, and J.~R. Hershey, ``Sparse {NMF} – half-baked or
  well done?'' Mitsubishi Electric Research Laboratories (MERL), Tech. Rep.,
  2015.

\bibitem{vincent}
E.~Vincent, R.~Gribonval, and C.~F\'{e}votte, ``Performance measurement in
  blind audio source separation,'' \emph{IEEE Transations on Audio, Speech, and
  Language Processing}, vol.~14, no.~4, pp. 1462 -- 1469, June 2006.

\bibitem{root3}
E.~Rechtschaffen, ``Real roots of cubics: explicit formula for
  quasi-solutions,'' \emph{The Mathematical Gazette}, no. 524, p. 268–276,
  2008.

\end{thebibliography}

\end{document}